\documentclass[conference]{IEEEtran}
\IEEEoverridecommandlockouts
\usepackage{cite}
\usepackage{amsmath,amssymb,amsfonts,amsthm}
\usepackage{algorithmic}
\usepackage{graphicx}
\usepackage{textcomp}
\usepackage{xcolor}
\usepackage{float}
\usepackage{placeins}  
\usepackage{setspace}
\usepackage{graphicx}
\usepackage[font=small]{caption} % 设置字体大小
\usepackage[skip=6pt]{caption}
\usepackage{bm}
\usepackage[font=footnotesize]{caption}
\usepackage{booktabs}
\usepackage{multirow}
\usepackage{colortbl}
\usepackage[hidelinks]{hyperref}

\hypersetup{
    colorlinks=false,
    pdfauthor={Mengxiang Liu, Xin Zhang, Shiyi Zhao, and Ruilong Deng},
    pdftitle={Hierarchical Detection and Mitigation Framework against Sensor Spoofing Attacks in Networked DC Microgrids},
    pdfsubject={Sensor spoofing attack detection and mitigation in networked microgrids},
    pdfkeywords={sensor spoofing attack, networked microgrids, proactive detection, impact mitigation, multi-layer coordination}
}

\def\BibTeX{{\rm B\kern-.05em{\sc i\kern-.025em b}\kern-.08em
    T\kern-.1667em\lower.7ex\hbox{E}\kern-.125emX}}
\begin{document}

\newtheorem{Propos}{Proposition}

\title{
% \fontsize{18}{24}\selectfont 
% Coding Matrix based Proactive Detection with Detectability-Hiddenness Optimisation in Cyber-Physical Microgrids
% Hierarchical Detection and Mitigation Framework against Sensor Spoofing Attacks in Networked {\color{black}DC} Microgrids
\color{black}{Hierarchical Sensor-Spoofing Defence Framework for Networked DC Microgrids via Cyber-Physical Coordination}
% Detection Using Meter Coding Scheme in Cyber-Physical Power Systems
\\
\thanks{
This work was supported in part by the U.K. Research and Innovation Future Leaders Fellowship ‘Digitalisation of Electrical Power and Energy Systems Operation’ under Grant MR/W011360/2. 
% in part by the U.K. Engineering and Physical Sciences Research Council (EPSRC) through Marie Sklodowska-Curie Actions (MSCA) Postdoctoral
% Fellowship CREDIT-NMG under Grant EP/Z533531/1, 
% and in part by the Royal Society International Exchanges under Grant IEC{\textbackslash}NSFC{\textbackslash}242449.
}
\thanks{$^{1}$School of Electrical and Electronic Engineering, The University of Sheffield, UK;
% }
% \thanks{
% $^{2}$ Department of Electrical and Electronic Engineering, Imperial College London, UK;
% }\thanks{
$^{2}$School of Computer Science, University of Bristol, UK.
$^{3}$Department of Control Science and Engineering, Zhejiang University, 
China.
(\textit{Corresponding Authors: Xin Zhang, Ruilong Deng})}
% University of Sheffield; R. Deng is with 
% The authors are with the School of Electronic and Electrical Engineering, University of Sheffiel
}

% \author{\IEEEauthorblockN{1\textsuperscript{st} Zhuoran Zhou}
% \IEEEauthorblockA{\textit{School of Electrical and Electronic Engineering} \\
% \textit{University of Sheffield}\\
% Sheffield, UK \\
% zzhou106@Sheffield.ac.uk}
% \and
% \and
% \and
% \and
% \and

% \author{Zhuoran Zhou, Mengxiang Liu, Xin Zhang}
\author{Mengxiang Liu$^{1,2}$, Xin Zhang$^{1}$, Shiyi Zhao$^{3}$, and Ruilong Deng$^{3}$}
% \vspace{-30pt}
\maketitle
\thispagestyle{plain}
\pagestyle{plain}
% \vspace{-30pt}
\begin{spacing}{1}
% \vspace{-30pt}
\begin{abstract}
% \vspace{-30pt}

In parallel to the cyber attack that manipulates the reference points of distributed energy resources (DERs) by maliciously accessing the remote monitoring and control system, the vulnerability of voltage/current sensors to electromagnetic interference (EMI) in the physical domain has been widely discussed.
Existing research efforts against sensor spoofing attacks can be classified into physical prevention and cyber detection/mitigation. 
These defence methods each have strengths and weaknesses in balancing cost, security, and performance in a single DER, yet systematic research on their multi-layer efficient coordination across DERs remains limited. Towards this end, this paper proposes a hierarchical framework to detect and mitigate sensor spoofing attacks in networked {\color{black}{DC}} microgrids (NMGs) via {multi-layer cyber-physical coordination}. It requires only to deploy physical prevention technologies at critical points, i.e., the local points of common coupling (PCC) of MGs, such that cyber detection/mitigation algorithms can be adopted based on the secured sensor readings to counter sensor spoofing attacks in DERs. The framework employs {an} MG-DER coordinated proactive detection scheme to strategically trigger parameter perturbations, under which the intelligent sensor spoofing attacks can be {effectively} disclosed.
Afterwards, mitigation schemes based on MG-DER coordination are activated to recursively and accurately estimate sensor biases.
Experiments on a cyber-physical {\color{black}{DC}} NMG testbed confirm the framework's effectiveness across diverse attack scenarios.
\end{abstract}


\section{Introduction}
The power grid is undergoing rapid digitalisation to accommodate the integration of distributed energy resources (DERs), improving {convenient energy management and cost-efficient operation}. However, this digital connection also exposes DERs to numerous cyber threats such as {unauthorised access to} the remote monitoring/management system to manipulate their reference points \cite{BH2025solar,Forescout2025}, which has stimulated extensive concerns in the cybersecurity and power system communities \cite{9737024,9583906,9832494}. Recently, the sensors within DERs are suspected to be vulnerable to electromagnetic interference (EMI), which could affect current and voltage readings by injecting it into the differential operational amplifier of {the} sensor \cite{barua2020hall}.
% , compared with the previous mentioned cyber manipulator, comes in an unconventional way to affect the voltage and current readings by emitting carefully and intentionally crafted EMI near the sensor \cite{barua2020hall}. 
This so-called sensor spoofing attack can lead to three primary consequences: denial-of-service, physical damage to power inverters, and a reduction in the power output of DERs \cite{yang2024rethink}. These attacks are particularly concerning in modern power grids, where the increasing penetration of DERs reduces physical inertia, indirectly making the power grid sensitive to attack disturbances
% potentially leading to fast-propagating grid collapse 
\cite{9796617,11029084}.

Numerous research efforts have been devoted to resisting sensor spoofing attacks, which, considering the attack's {physical-domain origin}, are divided into hardware-assisted sensor hardening and cyber detection/mitigation \cite{10459229}. Hardware-assisted sensor hardening includes matched-dummy transduction-shield circuits that detect and correct EMI-corrupted measurements \cite{tu2021transduction} and real-time in-sensor defence mechanisms against magnetic spoofing on Hall sensors \cite{barua2022halc}, which are representative sensor-side circuit hardening methods against EMI-induced measurement corruption \cite{5941844}. 
% Tu \textit{et al.} proposed a method to detect and correct malicious EMI injections using a matched dummy sensor circuit \cite{tu2021transduction}. This circuit mirrors the sensor’s vulnerability to injected EMI but remains unaffected by legitimate signals the sensor is designed to measure. To strengthen defences against EMI injections of varying frequencies and strengths, a real-time in-sensor defence mechanism was developed \cite{barua2022halc}. This mechanism employs two parallel cores: an analog core that filters out fake time-dependent magnetic fields using fast-order filters, and a digital core that eliminates fake constant fields with a DC feedback signal, preserving the original signal. 
However, although physical prevention technologies can effectively counter EMI injections, they often demand additional hardware modifications or installations, raising costs and extending {the implementation period}.

Cyber defence strategies use advanced data processing to analyse and model the dynamic behaviour of sensor measurements, enabling anomaly detection and impact mitigation against sensor spoofing attacks \cite{10836761}. Unlike hardware-assisted sensor hardening, which often requires sensor-side circuit modifications, cyber strategies mainly leverage model-based observers, robust estimation, and data-driven analysis to perceive and respond to sensor spoofing attacks. Pan \textit{et al.} modelled the grid-tied photovoltaic (PV) inverter as a linear time-varying system to account for changing environmental factors and unknown system states, applying robust estimation and control techniques to detect and mitigate sensor spoofing attacks \cite{10707330,10643338}. Peng \textit{et al.} integrated the physical dynamics of the grid-tied PV system into a graph neural network using a dynamic adjacency matrix, developing a learning-based method to detect sensor spoofing attacks with enhanced accuracy and interpretability \cite{10638139}. Zhang \textit{et al.} examined a PV farm aggregating multiple PV units at a point of common coupling (PCC), employing harmonic state space modelling to characterise state correlations in the frequency domain for attack detection \cite{9580468}.
More recently, in the context of networked {microgrids} (NMGs), where geographically close DERs are aggregated into multiple MGs for hierarchical control and operation, the intricate cyber-physical interdependences among DERs {have} stimulated extensive discussions. Unknown input observer (UIO) was leveraged to decouple the impacts of these interdependences, {supporting recursive detection and mitigation of} sensor spoofing attacks within DERs \cite{10746504}. Given the increasing sophistication of cyber threats, proactive detection methods were developed to resist intelligent adversaries {with certain model knowledge} by strategically perturbing control gains \cite{9621221} and adding physics-aware watermarks \cite{10643207}.

{\color{black}Despite the significant progress achieved in cyber-physical attack detection and mitigation, the research gaps become evident when representative studies are examined from the three aspects summarised in Table \ref{tab:introcomparison}. First, most existing cyber-layer defence methods are primarily designed for traditional false data injection attacks or model-level cyber attacks, while giving limited consideration to the physical-domain origin and propagation mechanism of sensor spoofing attacks \cite{10643207,10836761,9621221}. Consequently, matrix coding \cite{10836761}, control perturbation \cite{9621221}, and watermark embedding \cite{10643207} may be ineffective against spoofing biases that have already corrupted the sensor measurements before data transmission to the controller. Second, the absence of hierarchical detection makes it difficult to jointly balance detection performance and operational cost, particularly under stealthy attacks that can bypass local model-based detectors \cite{10638139,9580468,9621221,10643207}. Third, the lack of hierarchical mitigation prevents existing methods from fully exploiting the hierarchical control architecture of NMGs to counter adversaries, and some methods such as \cite{10746504} may require extra hardware for impact mitigation. Therefore, a coordinated framework is still needed to disclose stealthy sensor spoofing attacks and mitigate their impacts while remaining compatible with existing sensor-side prevention mechanisms.}
\begin{table}[!t]
    \centering
    \caption{{\color{black}Comparison of this paper with representative studies.}}
    \label{tab:introcomparison}
    {\color{black}
    \setlength{\tabcolsep}{2.2pt}
    \renewcommand{\arraystretch}{1.2}
    \resizebox{\linewidth}{!}{%
    \begin{tabular}{>{\centering\arraybackslash}m{4.4cm}>{\centering\arraybackslash}m{2.3cm}>{\centering\arraybackslash}m{2.3cm}>{\centering\arraybackslash}m{2.3cm}}
        \toprule[1pt]
        \textbf{\shortstack[c]{{References for}\\{cyber defence}\\{methods}}} & \textbf{\shortstack[c]{{Effectiveness}\\{against sensor}\\{spoofing}}} & \textbf{\shortstack[c]{Hierarchical\\detection\\~}} & \textbf{\shortstack[c]{Hierarchical\\mitigation\\~}} \\
        \midrule
        Matrix Coding \cite{10836761} & No & No & No \\
        {Robust estimation} \cite{10707330,10643338} & Yes & No & No \\
        Dynamic graph \cite{10638139} & Yes & No & No \\
        {Harmonic modelling} \cite{9580468} & Yes & Partial & No \\
        {Recursive method} \cite{10746504} & Yes & No & No \\
        Control Perturbation \cite{9621221} & No & Partial & No \\
        {Watermarking} \cite{10643207} & No & No & No \\
        \rowcolor{black!12}
        \textbf{This paper} & \textbf{Yes} & \textbf{Yes} & \textbf{Yes} \\
        \bottomrule[1pt]
    \end{tabular}}}
\end{table}
% to resist intelligent sensor spoofing attacks, 
In particular, hardware-assisted approaches can be integrated into cyber defence strategies, where physical prevention technologies are deployed at critical points to enable effective cyber detection and mitigation against sensor spoofing attacks on other points. Towards this end, we propose a hierarchical attack detection and impact mitigation framework for NMGs via cyber-physical coordination, requiring only to secure the sensors at the local PCC within {each MG} using physical prevention technologies. Based on the secured sensor information, detection and impact mitigation strategies are coordinated across MG and DERs to counter the sensor spoofing attacks within DERs. The contributions are summarised as follows:
\begin{itemize}
    \item We propose a hierarchical attack detection and impact mitigation framework to counter sensor spoofing attacks in NMGs via multi-layer cyber-physical coordination.
    % by strategically integrating MG and DER resources.
    \item We design a MG-DER coordinated proactive detection scheme to enhance the detection capability against intelligent sensor spoofing attacks, where the detection effectiveness and dynamic stability under parameter perturbation are analytically analysed.
    \item We establish impact mitigation strategies under diverse attack scenarios by integrating MG-DER layer resources, where the sensor biases can be recursively and accurately estimated without requiring hardware alterations.
    \item {\color{black}The performance of proposed framework under various perturbation strengths, parameter uncertainties, heterogeneous converters, and coordination delay is validated via extensive experimental studies.}
    % in a real-time power and communication co-simulated NMG.
\end{itemize}

\begin{figure*}
    \centering
    \includegraphics[width=1\linewidth]{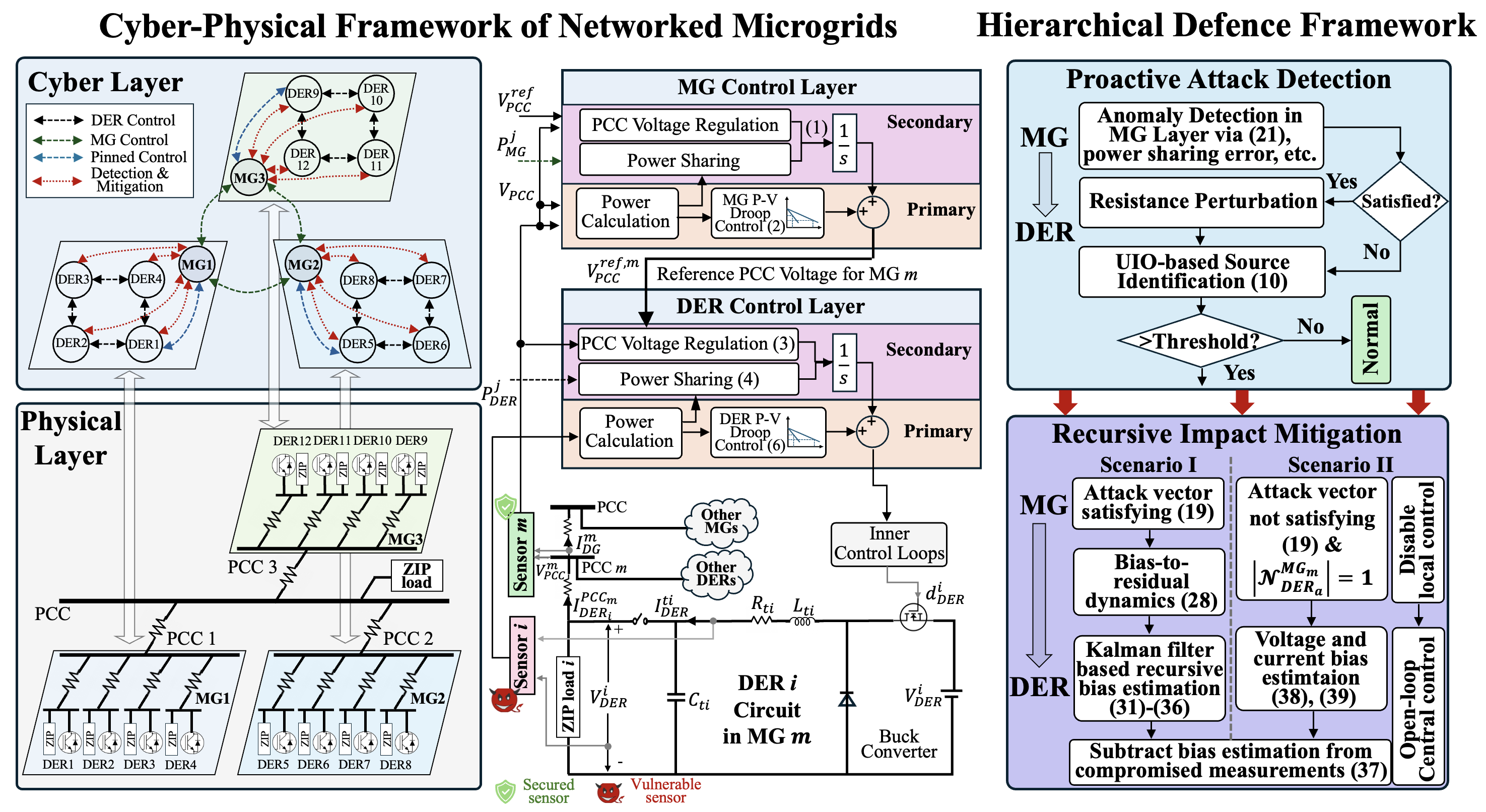}
    \caption{The figure shows the cyber-physical framework of NMGs and the proposed hierarchical defence framework, where the MG and DER layers are coordinated to proactively detect and recursively mitigate intelligent sensor spoofing attacks. For proactive detection, the anomaly observed in the MG layer is used as {a trigger} for the proactive resistance perturbation, effectively disclosing advanced sensor spoofing attacks. Then, {the information revealed by MG-layer verification is adopted by the DER layer to recursively estimate sensor biases.}
    }
    \label{fig:system-diagram}
\end{figure*}

\section{System Models}
This section introduces the {\color{black}DC} NMG system model, UIO-based detection model, and sensor spoofing attack model.

% the system model of single-bus NMG, and then details the proposed detection and mitigation framework.

\subsection{NMG System Model}
We consider a cyber-physical {\color{black}DC} NMG as shown in the left part of Fig. \ref{fig:system-diagram}, where, in the physical layer, DERs are connected to a local PCC in the MG and MGs are connected to a common PCC. Suppose that the NMG includes $N_{MG}$ number of MGs, represented by set $\mathcal{N}_{MG} = \{1, \cdots, N_{MG}\}$, and the set of DERs in MG $m\in\mathcal{N}_{MG}$ is denoted by $\mathcal{N}^{MG_m}_{DER} = \{1,\cdots,N^{MG_m}_{DER}\}$. The cyber layer consists of the DER and MG controllers as well as their intra-links and interactions with the physical circuit. 

As shown in the middle part of Fig. \ref{fig:system-diagram}, the hierarchical control scheme is adopted to accomplish the power sharing and PCC voltage regulation objectives within the MGs and NMG in a distributed manner. Following a top-to-bottom view, the control structure within MG controller $m \in \mathcal{N}_{MG}$ consists of secondary and primary controllers, where the secondary control input $\alpha^m_{MG}$ is governed by the following pinning consensus control scheme
{
\small
\begin{align}\label{eq: MG secondary}
    \dot{\alpha}_{MG}^m = 
    % \alpha_m(k-1) + 
    b_{MG}^{m}(V_{PCC}^{ref} - V_{PCC}) + \sum_{j\in\mathcal{N}_{MG}}a_{MG}^m(\frac{P_{MG}^j}{r_{MG}^j} - \frac{P_{MG}^m}{r_{MG}^m}),
\end{align}}where $V_{PCC}$ is the voltage measurement of common PCC, $P_{MG}^m$ is the output power of MG $m$, calculated as the product of PCC $m$'s voltage $V_{PCC_m}$ and output current $I_{MG}^m$, and $r_{MG}^m$ is the rated value of MG $m$'s output power. Parameters $b_{MG}^m$ and $a_{MG}^m$ are the {pinning} control gain and link weight between MGs $m$ and $j$, respectively. After integrating $\alpha^m_{MG}$, the reference voltage for PCC $m$, denoted by $V_{PCC}^{ref,m}$, is calculated from the subsequent {P--V droop-based} primary control as follows
\begin{align}\label{eq: MG primary control}
    V_{PCC}^{ref,m} = V_{MG}^{m*} - D_{MG}^mP_{MG}^m + \alpha_{MG}^{m},
\end{align}where $D_{MG}^m>0$ is the {droop control gain} of MG $m$ and $V_{MG}^{m*}$ is the rated output voltage of MG $m$. The reference PCC voltage $V_{PCC}^{ref,m}$ will be transmitted to the bottom DER controllers to achieve power sharing and PCC voltage regulation within MG $m$. In particular, the secondary controller of DER $i$ similarly adopts the pinning consensus control scheme, formulated as
{
% \small
\begin{align}
    % \text{Voltage:}\quad
    \dot{v}_{DER}^{i} &= b_{DER}^{i}(V_{PCC}^{ref,m} - V_{PCC}^m) + \nonumber \\
    & \quad + \sum_{j\in\mathcal{N}_{DER}^{MG_m}}a_{DER}^{ij,v}(v_{DER}^{j} - v_{DER}^{i})\\ 
    % \text{Power:}\quad
    \dot{p}_{DER}^{i} &= \sum_{j\in\mathcal{N}_{DER}^{MG_m}}a_{DER}^{ij,p}(\frac{P_{DER}^{j}}{r_{DER}^{j}} - \frac{P_{DER}^{i}}{r_{DER}^{i}}),
\end{align}}and the secondary control input is obtained by suming the above voltage ($v_{DER}^{i}$) and power ($p_{DER}^{i}$) related consensus variables as
\begin{align}
    \alpha_{DER}^{i} = v_{DER}^{i} + p_{DER}^{i},
\end{align}where $V_{PCC}^m$ is the voltage measurement of local PCC $m$ and $P_{DER}^{i}$ is the output power of DER $i$, calculated as the product of output voltage and current, i.e., $V_{DER}^{i}I_{DER}^{ti}$. Parameters $a_{DER}^{ij,v}$ and $a_{DER}^{ij,p}$ are the link weights, corresponding to the interaction of $v_{DER}^{j}$ and $p_{DER}^{j}$, respectively, between DERs $i$ and $j$, and parameter $b_{DER}^{i}$ is the pinning consensus control gain, being $1$ only if DER $i$ is selected as the pining node. Similarly, the reference value for DER $i$'s output voltage $V_{DER}^{ref,i}$ is then calculated from the P-V droop-based primary control, i.e.,
\begin{align}\label{eq: DER primary control}
    V_{DER}^{ref,i} = V_{DER}^{i*} - D^{i}_{DER}P_{DER}^{i} + \alpha_{DER}^{i},
\end{align}where $D^{i}_{DER}$ is the droop control gain of DER $i$'s control layer and $V_{DER}^{i*}$ is the rated output voltage of DER $i$. Based on voltage and current inner control loops, the DC-DC buck converter is regulated to track $V_{DER}^{ref,i}$, supplying local {constant impedance, constant current, and constant power (ZIP) loads} via a {resistor, inductor, and capacitor (RLC)} filter.

Based on the Kirchhoff current and voltage laws, the average dynamical model of a DC-DC buck converter in DER $i$ can be formulated as a linear time-invariant state-space model according to \cite{tan2015dc} as follows
\begin{align}\label{eq: DER SS Model}
\left\{
  \begin{array}{ll}
    \dot{\bm{x}}_i = A_i \bm{x}_i + B_iu_i + E_ig_i + \bm{\omega}_i \\
    \bm{y}_i = \bm{x}_i + \bm{\rho}_i
    \end{array}
\right.,
\end{align}where the system state vector $\bm{x}_i = [V_{DER}^{i}, I_{DER_{}}^{ti}]^{\rm T}$, system input $u_i = d_{DER}^{i}V_{RES}^{i}$, and $d_{DER}^{i}$ is the duty cycle obtained from the inner control loop. {The constant voltage source} $V_{RES}^{i}$ simplifies renewable energy sources (RESs) with sufficient storage capacity, and exogenous input $g_i = I_{DER}^{ZIP,i} + I_{DER_{i}}^{PCC_m}$ synthesises the current component of ZIP load within DER $i$ and the current flow from DER $i$ to PCC $m$. Moreover, $|\bm{\omega}_i| \le \bar{\bm{\omega}}_i$ and $|\bm{\rho}_i| \le \bar{\bm{\rho}}_i$ represent the process and measurement noises with known bounds, respectively, in the evolving dynamics. The system parameters $A_i, B_i$, and $E_i$ are expanded as
\begin{align}\label{eq: system parameters}
    A_i &= \begin{bmatrix}
-\frac{1}{Z_{DER}^{ZIP,i}C_{ti}} & \frac{1}{C_{ti}} \\
-\frac{1}{L_{ti}} & -\frac{R_{ti}}{L_{ti}}
\end{bmatrix}, 
B_i = \begin{bmatrix} 
0 \\
\frac{1}{L_{ti}},
\end{bmatrix},
E_i = \begin{bmatrix} 
-\frac{1}{C_{ti}} \\
0
\end{bmatrix},
\end{align}where $Z_{DER^{}}^{ZIP,i}$ is the resistance component of the ZIP load within DER $i$ and $R_{ti}, L_{ti}, C_{ti}$ are the RLC filter parameters. The continuous-time form \eqref{eq: DER SS Model} can be transformed into the following discrete-time form under sampling period $T_{samp}$
{\small\begin{align}\label{eq: Discrete DER SS Model}
\left\{
  \begin{array}{ll}
    {\bm{x}}_i(k+1) = A_{di} \bm{x}_i(k) + B_{di}u_i(k) + E_{di}g_i(k) + \bm{\omega}_i(k) \\
    \bm{y}_i(k) = \bm{x}_i(k) + \bm{\rho}_i(k)
    \end{array}
\right.,
\end{align}}where the discrete-time system parameters are derived from the continuous-time ones, i.e., $A_{di} = e^{A_iT_{samp}}$, $B_{di} = \int_{\tau = 0}^{T_{samp}}e^{A_i\tau}d\tau B_i$, and $E_{di} = \int_{\tau = 0}^{T_{samp}}e^{A_i\tau}d\tau E_i$. 

\subsection{UIO-based Attack Detection}
Since not all DERs have access to the exogenous input $g_i$, it is usually regarded as unknown, thus necessitating the design of UIO-based attack detector \cite{chen1996design}. To guarantee that the impact of unknown input $g_i$ can be {decoupled} when observing state $\bm{x}_i$ from measurement $\bm{y}_i$, the sufficient and necessary condition is that the transmission zeros from $g_i$ to $\bm{y}_i$ are stable, which is always satisfied as matrix $\begin{bmatrix} 
sI - A_{di} & E_{di} \\
I & 0
\end{bmatrix}$ is of full column rank for all $s$ with ${\rm Re}(s)\ge 0$ according to \eqref{eq: system parameters}. Then, the designed UIO follows
{
% \small
\begin{align}\label{eq: UIO model}
\left\{
  \begin{array}{ll}
    {\bm{z}}_i(k+1) = F_{di} \bm{z}_i(k) + T_{di}B_{di}u_i(k) + \hat{K}_{di}\bm{y}_i(k) \\
    \hat{\bm{x}}_i(k) = \bm{z}_i(k) + H_{di}\bm{y}_i(k)
    \end{array}
\right.,
\end{align}}where $\bm{z}_i$ is the internal UIO state and UIO parameters $F_{di}, T_{di}, \hat{K}_{di}$, and $H_{di}$ are selected according to
\begin{align}
    &(H_{di} - I)E_{di} = \bm{0} \label{eq: UIO parameter 1} \\
    &T_{di} = I - H_{di} \\
    &F_{di} = T_{di}A_{di} - \hat{K}_{di} + F_{di}H_{di}.
\end{align} 

According to the DER dynamics \eqref{eq: Discrete DER SS Model} and UIO model \eqref{eq: UIO model}, the residual $\bm{r}_i = \bm{y}_i - \hat{\bm{x}}_i$ that represents the discrepancy between the actual measurement and estimated state is derived as
\begin{align}\label{eq: detection residual}
    \bm{r}_i(k) &= (F_{di})^k \big(\bm{r}_i(0) - T_{di}\bm{\rho}_i(0)\big) + T_{di}\bm{\rho}_i(k) + \nonumber \\
    & + \sum_{l=0}^{k-1}(F_{di})^{k-l-1}\big(T_{di}\bm{\omega}_i(l) - \hat{K}_{di}\bm{\rho}_i(l)\big),
\end{align}which will asymptotically converge to a range close to $0$ when the eigenvalues of $F_i$ are within the unit circle. In normal cases, it is easy to find its upper bound to establish
\begin{align}\label{eq: upper bound}
    |\bm{r}_i(k)| \le \bar{\bm{r}}_i(k) &= \nu_i(\varsigma_i)^k\big|T_{di}\big|\bar{\bm{\rho}}_i+|T_{di}|\bar{\bm{\rho}}_i \nonumber\\
    &\sum_{l=0}^{k-1}\nu_i(\varsigma_i)^{k-1-l}\big(|T_{di}|\bar{\bm{\omega}}_i+|\hat{K}_{di}|\bar{\bm{\rho}}_i\big),
\end{align}where positive scalars $\nu_i$ and $0<\varsigma_i<1$ are chosen such that $||(F_{di})^k|| \le \nu_i(\varsigma_i)^k$. Once the breach of \eqref{eq: upper bound} is observed, i.e., 
\begin{align}\label{eq: anomaly detection}
    |r_{i,V}| > \bar{r}_{i,V} \quad {\rm or} \quad |r_{i,I}| > \bar{r}_{i,I},
\end{align}the alarm of data anomaly will be triggered. The above subscripts {$V$ and $I$} denote the voltage and current related residual variables, respectively.

\begin{figure}[!t]
    \centering
    \includegraphics[width=1\linewidth]{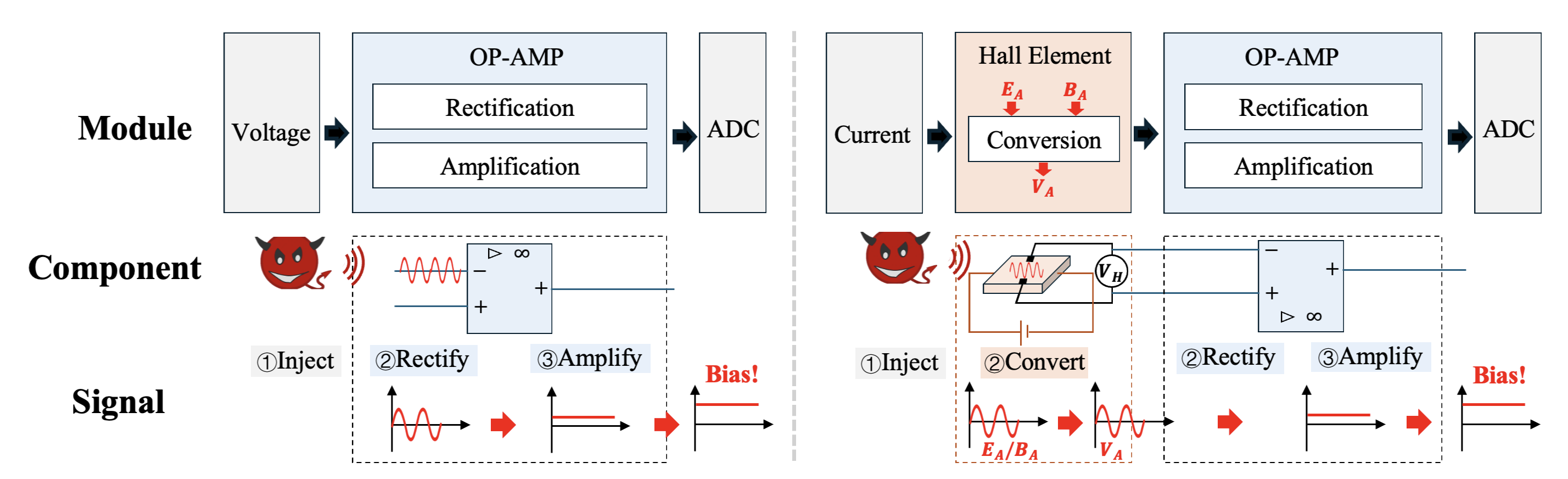}
    \caption{This figure illustrates the mechanisms of spoofing attacks against voltage and current sensors in DERs \cite{yang2024rethink}. {\color{black}
    % In the left part, the parasitic capacitance carried on the sensor's {printed circuit board (PCB)} couples the high-frequency electric fields resulting from EMI injections into the differential operational amplifier (OP-AMP). Then, the coupled interference is rectified and amplified by OP-AMP, eventually outputted as positive or negative bias on the voltage reading. In the right part, the current sensor includes not only the OP-AMP circuit but also a Hall element, which provides a new entrance for EMI injections. Specifically, the induced magnetic ($B_A$) or electric ($E_A$) field around the Hall chip is converted as voltage deviation $V_A$, which, after the process of OP-AMP circuit, is outputted as positive or negative bias on the current reading.
    }}
    \label{fig:sensorspoofing}
\end{figure}

\subsection{Sensor Spoofing Attack Model}
As shown in Fig. \ref{fig:sensorspoofing}, both voltage and current sensors within DERs are vulnerable to EMI injections due to the the electromagnetic sensitive components (parasitic capacitance) and process (Hall element) \cite{yang2024rethink}. {\color{black}The magnetic and electric fields resulting from EMI injections will be coupled into the sensor circuit, eventually outputted as positive or negative biases on sensor readings. 
Hall effect sensors have been widely adopted in converters and inverters for AC/DC voltage/current measurements, due to their guaranteed efficiency and accuracy under varying environments.
% \cite{ramsden2011hall}. 
The lack of security considerations in designing these sensors exposes grid-tied DERs to unconventional security concerns such as the carefully crafted EMI injection.
Therefore, the attacker is able to spoof the voltage/current sensor by crafting an intentional EMI and placing it near the sensor in a stealthy way, such as being camouflaged within a flower vase or coffee cup \cite{barua2020hall}.} 
Mathematically, 
% This paper considers the sensor spoofing attack that tamper with the local voltage and current sensor readings in the means of placing external magnetic field or compromising sensor firmware.
the sensor spoofing attack can be modelled as
\begin{align}
    \text{Voltage:}\quad V^{ia}_{DER} = V^i_{DER} + \phi^{V,i}_{DER}, \label{eq: spoofing attack V} \\ 
    \text{Current:}\quad I^{tia}_{DER} = I^{ti}_{DER} + \phi^{I,ti}_{DER}, \label{eq: spoofing attack I}
\end{align}where vector $\bm{\phi}_{DER}^i = [\phi_{DER}^{V,i}, \phi_{DER}^{I,ti}]^{\rm T}$ includes the spoofing biases injected into the voltage and current measurements. In addition to the measurement disruption capability, we assume that the attacker is intelligent enough and can spoof voltage and current sensors coordinately to mimic the DER dynamics \eqref{eq: Discrete DER SS Model}. 
% This objective can be achieved 
After obtaining system parameters $A_{di}, E_{di}$, the injected biases can be faked as ones resulting from the unknown input, i.e.,
\begin{align}\label{eq: stealthy attack vector}
    \bm{\phi}_{DER}^i (k+1) = A_{di}\bm{\phi}_{DER}^i(k) + E_{di} g_i^a(k), 
\end{align}where $g_i^a$ denotes the faked unknown input that is indistinguishable from $g_i$. The bias vector designed through \eqref{eq: stealthy attack vector}, along with the initial condition $\bm{\phi}_{DER}^i (k_a) = 0$, can easily bypass the UIO-based attack detector \eqref{eq: UIO model} without altering the detection residual. Mathematically, let $\bm{r}_i^a$ represent the residual under the spoofing attack; then {it follows the same detector residual dynamics as the normal case and remains within the detection bound, rather than being exactly zero in all practical noisy conditions} \cite{9793599}. 
% Although the adversary is assumed to have full knowledge of the DER electrical parameters, the coordination of sensor spoofing attacks against multiple DERs is difficult due to the

\section{Hierarchical Attack Detection and Impact Mitigation Framework}
As shown in the right part of Fig. \ref{fig:system-diagram}, it is considered that hardware shielding \cite{tu2021transduction} and filtering \cite{barua2022halc} have been deployed at the local PCC in MG to prevent intentional EMI injections from affecting these critical sensor readings. Based on {this secured information}, hierarchical defence strategies are designed with MG–DER coordination to mitigate sensor spoofing attacks. For the cross-layer proactive detection,
the MG layer anomaly detection is essentially a centralised verification for the data received from DER members by utilising the secured local PCC measurement and Kirchhoff current law.
% , and ZIP load information. 
Once {an} anomaly is observed, the RLC parameters within DERs will be strategically perturbed to cause {attack-defence information asymmetry}. Therefore, the UIO-based detector deployed in DERs could be able to perceive the carefully designed bias vectors \eqref{eq: stealthy attack vector}, 
% to expose the carefully crafted bias vector to the UIO-based detector, 
thus identifying the compromised DERs. The impact mitigation is discussed in three cases: 1) If the attack is detected after parameter perturbation, then a Kalman filter is adopted to recursively estimate the bias vector based on residual dynamics \eqref{eq: explicit residual bias relations} and attack generation scheme \eqref{eq: stealthy attack vector}; 2) When the attack is detected without perturbing parameters and the number of compromised DERs in a MG is equal to $1$, its legitimate voltage information is estimated from the Kirchhoff current law \eqref{eq: voltage reconstruction}, from which the current bias can be recursively estimated according to residual dynamics as \eqref{eq: recursive reconstruction form}; {3) When the attack is detected without perturbing parameters and the number of compromised DERs in a MG is more than $1$, the recursive estimation of sensor biases is impossible. To prevent the attack impacts from getting worse, the local control within these affected DERs will be disabled and the MG controller will take control of them.
% To mitigate and eliminate the attack impacts, 
}

% the identified malicious DERs will be selectively tripped to keep the number of compromised DERs at each local PCC not exceeding one. Then, the legitimate voltage and current measurements can be reconstructed from the Kirchhoff current law (MG layer) and UIO residual (DER layer), respectively, eventually achieving impact mitigation.

\subsection{Proactive Attack Detection}
As shown in the right upper part of Fig. \ref{fig:system-diagram}, the proactive attack detection synthesises the information at both MG and DER layers, corresponding to the centralised anomaly detection and decentralised source identification, respectively. The anomaly detected from the centralised MG viewpoint will trigger the perturbation on DER parameters within this MG such that the compromised DERs can be successfully identified.
% consists of 1) Current flow law based attack detection in the MG layer, and 2) Strategic parameter perturbation for attack identification within DERs, which are detailed in the subsequent part.

\subsubsection{Anomaly detection in the MG layer} 
To calculate the output power of MG $m$ and accomplish hierarchical control \eqref{eq: MG secondary}, \eqref{eq: MG primary control}, MG $m$ needs to measure PCC $m$'s voltage $V_{PCC}^m$ and MG's output current $I_{MG}^m$. These information, along with the DERs' output voltages, $V_{DER}^i, i \in \mathcal{N}_{DER}^{MG_m}$, motivates the Kirchhoff current law based verification at PCC $m$, i.e.,
{
% \small
\begin{align}\label{eq: current flow based verification}
    \tau_{PCC}^m &= I_{MG}^m - 
    % f^{ZIP}_{PCC_m}(V_{PCC}^m) + \nonumber \\ &\quad 
      \sum_{j\in\mathcal{N}_{DER}^{MG_m}}G_{DER_j}^{PCC_m}(V^j_{DER} - V_{PCC}^m),
\end{align}}where $\tau_{PCC}^m$ denotes the discrepancy between flowing-out and flowing-in currents at PCC $m$ 
% , function $f^{ZIP}_{PCC_m}(V_{PCC}^m) = V_{PCC}^m/Z_{PCC_m}^{ZIP} + I_{PCC_m}^{ZIP} + P_{PCC_m}^{ZIP}/V_{PCC}^m$ estimates the current flow to the ZIP load at PCC $m$, characterised by $Z_{PCC_m}^{ZIP}, I_{PCC_m}^{ZIP}, P_{PCC_m}^{ZIP}$. 
and $G^{PCC_m}_{DER_j}$ is the conductance of the power line connecting DER $j$ and PCC $m$. 
% The accuracy of estimating ZIP load current via $f_{PCC_m}^{ZIP}(\cdot)$ is guaranteed by advanced learning-driven load forecasting algorithms \cite{8743433}. 
To tolerate the measurement errors, a threshold $\bar{\tau}_{PCC}^m$ is predefined to bound $\tau_{PCC}^m$ in normal cases, i.e., $|\tau_{PCC}^m|\le \bar{\tau}_{PCC}^m$. Once
\begin{align}\label{eq: condition breach}
    |\tau_{PCC}^m| > \bar{\tau}_{PCC}^m,
\end{align}it is assumed that the received voltage $V_{DER}^j$ does not align with the Kirchhoff current law at PCC $m$, indicating that there might be issues with the received $V_{DER}^j, \forall j \in \mathcal{N}_{DER}^{MG_m}$ from DERs. 

{In extreme cases, when the adversary has full access to line parameter $G_{DER_j}^{PCC_m}$, then it is possible for the adversary to design {coordinated} bias injections such that $\sum_{j\in\mathcal{N}_{DER_a}^{MG_m}}G_{DER_j}^{PCC_m}\phi_{DER}^{V,j}
= 0$, where $\mathcal{N}_{DER_a}^{MG_m}$ is {the} set of compromised DERs. The {coordinated} bias injections can lead to {balanced} alterations of current flows, thus not varying $\tau_{PCC}^m$ and not enabling \eqref{eq: condition breach}. To address this limitation, apart from establishing \eqref{eq: condition breach}, other abnormal system behaviours such as large steady-state power sharing errors and swiftly varying output powers/voltages will also trigger the subsequent source identification in the DER layer.
% resistance perturbation.
% the periodical perturbation scheme will be established to handle the coordinated sensor spoofing attacks.
}

\subsubsection{Source identification in the DER layer} To identify the source of compromised DERs, the idea of proactive detection is adopted to expose the bias vector $\bm{\phi}_{DER}^i$ \eqref{eq: stealthy attack vector} that mimics the DER dynamics \eqref{eq: Discrete DER SS Model} to UIO-based detector \eqref{eq: UIO model}. The basic idea here is to perturb electrical parameters such as $R_{ti}, L_{ti}, C_{ti}$ to induce {attack-defence information asymmetry}, eventually increasing the UIO residual when the adversary uses outdated electrical parameters. Therefore, the sensitivity of UIO residual $\bm{r}_i$ with respect to parameter alterations {is} theoretically analysed to guide the parameter selection. 
\begin{Propos}\label{propos:residual sensitivity}
    Given that the sampling period $T_{samp}$ is small enough, the sensitivity of $\bm{r}_i^{\phi}$ with respect to the alterations of $R_{ti}, L_{ti}, C_{ti}, Z_{DER}^{ZIP,i}$ satisfy
    \begin{align}
        \frac{\partial \bm{r}_i^{\phi}}{\partial C_{ti}} &= \frac{\partial \bm{r}_i^{\phi}}{\partial Z_{DER}^{ZIP,i}} = 0 \label{eq: zero sensitivity} \\
        \Big|\frac{\partial \bm{r}_i^{\phi}}{\partial L_{ti}}\Big| &\varpropto \frac{R_{ti} + 1}{L_{ti}^2}|\bm{t}_{i2}|, \label{eq: inductance sensitivity} \\
        \Big|\frac{\partial \bm{r}_i^{\phi}}{\partial R_{ti}}\Big| &\varpropto \frac{1}{L_{ti}}|\bm{t}_{i2}|, \label{eq: resistance sensitivity}
    \end{align}where $\bm{r}_i^{\phi}$ denotes the residual variation when injecting bias vector $\bm{\phi}_i$ without recognising the parameter alterations, and $\bm{t}_{i2}$ is the second column vector of continuous-time UIO parameter $T_i$.
\end{Propos}

The proof of Proposition \ref{propos:residual sensitivity} is 
% omitted due to space limitation and will be provided as required.
provided in {appendix \ref{appendix: proof of propos}} of supplementary material. 
Proposition \ref{propos:residual sensitivity} validates that there might emerge residual variations only when perturbing electrical parameters $R_{ti}, L_{ti}$, potentially exposing the carefully designed bias vector \eqref{eq: stealthy attack vector} to the UIO-based detector. In particular, analysis in \eqref{eq: residual approximation} shows that the residual variation responds more strongly to changes in inductance $L_{ti}$ than to changes in resistance $R_{ti}$. {This paper nevertheless adopts resistance perturbation as the practical implementation target because it can be realised by a retrofit impedance-selection branch at the converter output stage, whereas inductance perturbation normally requires modifying the converter magnetic component or adding switchable inductive branches \cite{9069415}. The adopted resistance perturbation is therefore not claimed to be the most sensitive option; rather, it provides a physically simple and experimentally verifiable way to create the required attack-defence information asymmetry.}

Next, the stability of system matrix $\widetilde{A}_i$ under resistance perturbation is analysed. Based on \eqref{eq: system parameters}, the eigenvalues can be explicitly calculated as
\begin{align}\label{eq: eigenvalues}
    \lambda_i^{1,2} = \frac{a_{11} + a_{22}}{2} \pm \frac{\sqrt{(a_{11} - a_{22})^2 + 4a_{12}a_{21}}}{2},
\end{align}where $a_{11} = -\frac{1}{Z_{DER}^{ZIP,i}C_{ti}}, a_{12} = \frac{1}{C_{ti}}, a_{21} = -\frac{1}{L_{ti}}$ are the three entries of matrix $\widetilde{A}_{i}$ that are independent of $R_{ti}$ and $a_{22} = -\frac{\widetilde{R}_{ti}}{L_{ti}}$ is the affected element by the alteration of resistance. Considering that $a_{11}<0$, $a_{22}<0$, and $a_{11}a_{22}-a_{12}a_{21}>0$, {the characteristic polynomial has positive coefficients, i.e.,}
{\begin{align}
    \lambda^2-(a_{11}+a_{22})\lambda+(a_{11}a_{22}-a_{12}a_{21})=0,
\end{align}}{which verifies that both real eigenvalues are negative or, when the roots form a complex-conjugate pair, their real part $(a_{11}+a_{22})/2$ lies in the left half plane. Therefore, $\widetilde{A}_i$ remains stable under resistance perturbation.}

Therefore, it is concluded that the resistance perturbation can amplify the detection residual in the presence of deceptive sensor spoofing attacks \eqref{eq: stealthy attack vector}, while assuring the stability of circuit dynamics in steady states. Note that there should be a limit on the perturbation frequency to prevent the perturbation from happening when steady states are not attained from the last perturbation. However, it is difficult to determine the {optimal perturbed resistance} such that the attack can be detected with the minimal power loss increase as the attack design variable $g_i^a$ is not known in prior. {\color{black}{Based on the experimental results in subsection \ref{sub:scenarioI}, a larger resistance perturbation generally improves the UIO residual amplification and the subsequent bias-estimation accuracy, but it also increases the power line loss. Therefore, a practical selection rule is to increase the resistance perturbation according to the available retrofit resistance branches, e.g., $R_{ti}\to 5R_{ti}\to 10R_{ti}$, until the detection and mitigation performance becomes acceptable. A rigorous determination of the optimal perturbation branch can be formulated as an optimisation problem that jointly considers detection enhancement, bias-estimation accuracy, power-loss-induced economic burden, and the available resistance branches of the chosen retrofit module, which will be investigated in our future work.}}

\subsection{Recursive Impact Mitigation}
After identifying the compromised DER sources, the following step is to {mitigate} the attack impact. Two attack scenarios are considered: 1) When the sensor bias vector $\bm{\phi}_{DER}^i$ is generated from \eqref{eq: stealthy attack vector}, the proposed strategy {recursively estimates the biases} based on the detection residual $\bm{r}_i^{\phi}$ and bias generation scheme \eqref{eq: stealthy attack vector}. {The bias state vector consists of the bias vector $\bm{\phi}_{DER}^i$ and the faked unknown input $g_i^a$. The residual function $\bm{r}_i^{\phi}$ is then treated as the measurement output, and the attack impact can be mitigated by estimating the bias state vector from this output via a Kalman filter observer.} 2) The second scenario is open to the more general ones where the sensor spoofing attack is detected without triggering the resistance perturbation. Herein the generated bias vector $\bm{\phi}_{DER}^i$ does not conform dynamics \eqref{eq: stealthy attack vector} and thus the impact mitigation requires extra information from the MG layer. If the number of compromised DERs in a MG is equal to $1$, then the Kirchhoff current law \eqref{eq: current flow based verification} can be used to estimate the legitimate voltage sensor reading of affected DER, eventually enabling the recursive estimation of legitimate current reading. When the number of compromised DERs in a MG is more than $1$, the recursive estimation of legitimate sensor readings would be impossible. In this case, the closed-loop control within compromised DERs will be disabled and the MG controller is going to take over the local controller for regulating the PWM duty ratio.

\subsubsection{\texorpdfstring{Scenario I: Sensor Spoofing Attack Generated from \eqref{eq: stealthy attack vector}}{Scenario I: Sensor Spoofing Attack Generated from the Stealthy Attack Vector}}
According to \eqref{eq: Discrete DER SS Model}, \eqref{eq: UIO model}, when neglecting system noises, there exist explicit relations between residual $\bm{r}_i^{\phi}$ and injected bias vector $\bm{\phi}^i_{DER}$, i.e., 
{\small\begin{align}\label{eq: explicit residual bias relations}
    \bm{r}_i^{\phi}(k+1) = \widetilde{F}_{di}\bm{r}_i^{\phi}(k) + \widetilde{T}_{di}\bm{\phi}_{DER}^i(k+1) - \widetilde{T}_{di}\widetilde{A}_{di}\bm{\phi}_{DER}^i(k),
\end{align}}where the upper sign $\widetilde{\cdot}$ denotes the parameter after resistance perturbation. Assuming that the faked unknown input $g_i^a$ used by the attacker varies slowly and can be described by $g_i^a(k+1) = g_i^a(k) + \omega_i^g(k)$, where $\omega_i^g(k)$ is the normal distribution noise. Then, the bias state vector can be created as $\bm{\psi}_i = [(\bm{\phi}_{DER}^i)^{\rm T}, g_i^a]^{\rm T}$, and let the observation output be $\bm{d}_i(k) = \bm{r}_i^{\phi}(k+1) - \widetilde{F}_{di}\bm{r}_i^{\phi}(k)$, thus the related dynamics are 
% described by
{
% \small
\begin{align}\label{eq: scenario I estimation}
\left\{
  \begin{array}{ll}
    \bm{\psi}_i(k+1) = F_{di}\bm{\psi}_i(k) + \bm{\omega}_i^{\psi}(k) \\
    \bm{d}_i(k) = H_{di}\bm{\psi}_i(k) + \bm{\rho}_i^{\psi}(k)
    \end{array}
\right.,
\end{align}}where the matrixes $F_{di}, H_{di}$ are
\begin{align}\label{eq: EstimationParameter}
    F_{di} &= \begin{bmatrix}
A_{di} & E_{di} \\
\bm{0}^{\rm T} & 1
\end{bmatrix}, 
H_{di} = \begin{bmatrix} 
\widetilde{T}_{di}(A_{di} - \widetilde{A}_{di}) & \widetilde{T}_{di}E_{di}
\end{bmatrix},
\end{align}and the system noises follow $\bm{\omega}_i^{\psi} \sim \mathcal{N}(0,Q), \bm{\rho}_i^{\psi} \sim \mathcal{N}(0, R)$. The observability of matrix pair $(F_{di}, H_{di})$ is the same as its continuous form $(F_i, H_i)$ if the imaginary difference of conjugate eigenvalues ($\pm i\zeta$) of $A_{i}$ does not lead to 
\begin{align}
\zeta T_{samp} \neq 2\pi m, m \in \mathbb{Z}.
\end{align}
According to \eqref{eq: eigenvalues}, the two eigenvalues will become conjugate as the increase of resistance ${R}_{ti}$. Nevertheless, given the millisecond sampling time $T_{samp}$, it is impractical to alter $R_{ti}$ to k$\Omega$-level. Therefore, the matrix pair $(F_{di}, H_{di})$ is observable akin to its continuous form $(F_i, H_i)$ as proved in {appendix \ref{appendix: observability proof}} of supplementary material.

The Kalman filter is adopted to achieve the accurate estimation of bias vectors while keeping robustness to system noises. When the malicious DERs are perceived at time $k_a$, the initial bias state vector is set as $\widehat{\bm{\psi}}_i^{k_a|k_a} = \bm{0}$ and the initial covariance $P_i^{k_a|k_a}$ should be large enough to indicate large initial uncertainties. Then, the recursive estimation process at $k\ge k_a$ follows
{\small\begin{align}
&\text{Prediction}: \widehat{\bm{\psi}}_i^{k+1|k} = F_{di} \widehat{\bm{\psi}}_i^{k|k}, P_i^{k+1|k} = F_{di}P_i^{k|k}F_{di}^{\rm T} + Q, \\
&\text{Innovation residual}: \bm{\delta}_i^{k+1} = \bm{d}_i^{k+1} - H_{di}\widehat{\bm{\psi}}_i^{k+1|k},\\
& \text{Innovation covariance}: S_i^{k+1} = H_{di}P_i^{k+1|k}H_{di}^{\rm T} + R, \\
& \text{Kalman gain}: K_i^{k+1} = P_i^{k+1|k}H_{di}^{\rm T}(S_i^{k+1})^{-1}, \\
& \text{Update estimate}: \widehat{\bm{\psi}}_i^{k+1|k+1} = \widehat{\bm{\psi}}_i^{k+1|k} + K_i^{k+1}\bm{\delta}_i^{k+1}, \\
& \text{Update covariance}: P_i^{k+1|k+1} = (I-K_i^{k+1}H_{di})P_i^{k+1|k}.
\end{align}}It is noted that the time label is put in the right upper of variable for space saving. The variance of output noise, resulting from the original system noises in \eqref{eq: Discrete DER SS Model}, is characterised by diagonal matrix $R$. The variance of process noise is characterised by diagonal matrix $Q$, where the third diagonal entry should be larger than the former two entries to capture the dynamics of $g_i^a$. The increase of noise variance results in faster tracking of $g_i^a$ but noisier estimates. 

The estimated sensor biases are directly {subtracted} from the compromised measurement vector $\bm{y}_i^a$ to achieve impact mitigation:
\begin{align}\label{eq: impact mitigation}
    \bm{y}_i^{rec}(k) = \bm{y}_i^a(k) - \Phi\widehat{\bm{\psi}}_i(k), k\ge k_a,
\end{align}where $\bm{y}_i^{rec}$ denotes the recovered measurements incorporated into control loops and the selection matrix $\Phi = {\rm diag}([1,1,0])$.

\subsubsection{\texorpdfstring{Scenario II: Sensor Spoofing Attack not Satisfying \eqref{eq: stealthy attack vector}}{Scenario II: Sensor Spoofing Attack not Satisfying the Stealthy Attack Vector}}
When the sensor spoofing attack does not satisfy the deceptive generation scheme \eqref{eq: stealthy attack vector}, there would be no sufficient measurements to observe the sensor bias vector. If the number of affected DERs in the MG equalise to $1$, then the Kirchhoff current law at PCC $m$ \eqref{eq: current flow based verification} can be used to estimate its legitimate voltage information as
{\begin{align}\label{eq: voltage reconstruction}
    \widehat{V}_{DER}^{i} &= V_{PCC}^m + R_{DER_i}^{PCC_m} \Big(I_{MG}^m +
    % + f_{PCC_m}^{ZIP}(V_{PCC}^m) + 
    \nonumber \\ 
    &
    - \sum_{j\in\mathcal{N}_{DER}^{MG_m} / \mathcal{N}_{DER_a}^{MG_m}}G_{DER_j}^{PCC_m}\big(V_{DER}^j - V_{PCC}^m\big)\Big),
\end{align}}where $R_{DER_i}^{PCC_m}$ denotes the resistance of power line connecting DER $i$ and PCC $m$. After reconstructing the voltage information $V_{DER}^{i,rec}$ in the MG layer, it will then be transmitted to DER $i$ to complete the reconstruction of legitimate current information. Then, the voltage bias injection can be assumed known from $\widehat\phi_{DER}^{V,i} = V_{DER}^{ia} - \widehat{V}_{DER}^{i}$. According to \eqref{eq: explicit residual bias relations}, the recursive estimation scheme is formulated as
\begin{align}\label{eq: recursive reconstruction form}
    &\widehat{\phi}_{DER}^{I,i} (k+1) = (\bm{t}_{di2}^{nor})^{\rm T}T_{di}A_{di}\widehat{\bm{\phi}}_{DER}^{i}(k) + \nonumber \\
    & - (\bm{t}_{di2}^{nor})^{\rm T}\bm{t}_{di1}\widehat{\phi}_{DER}^{V,i} +  (\bm{t}_{di2}^{nor})^{\rm T}(\bm{r}_i^{\phi}(k+1) - F_{di}\bm{r}_i^{\phi}),
\end{align}where $\widehat{\bm{\phi}}_{DER}^{i} = [\widehat{{\phi}}_{DER}^{V,i}, \widehat{{\phi}}_{DER}^{I,ti}]^{\rm T}$ denotes the estimated bias vector of $\bm{\phi}_{DER}^i$, $\bm{t}_{di2}^{nor} = \frac{\bm{t}_{di2}}{(\bm{t}_{di2})^{\rm T}\bm{t}_{di2}}$, and $T_{di} = [\bm{t}_{di1}, \bm{t}_{di2}]$. Note here that the altered parameters are not used as the sensor spoofing attack not satisfying \eqref{eq: stealthy attack vector} can be easily detected without triggering perturbation. The recursive estimation scheme \eqref{eq: recursive reconstruction form} is proved to be stable under bounded system noises as well as uncertainties of initial states, electrical parameters, and renewable intermittency as the first entry of row vector $(\bm{t}_{di2}^{nor})^{\rm T}T_{di}A_{di}$ always stays within the unit circle. Eventually, the estimated bias vector $\widehat{\bm{\phi}}_{DER}^{i}$ will be subtracted from the affected measurements to accomplish impact mitigation similar to \eqref{eq: impact mitigation}.

If the number of compromised DERs in a MG is more than $1$, it would be essentially difficult to recursively estimate the sensor biases due to the loss of observability. In {this} case, the closed-loop control within DERs is not reliable and will be disabled. {The MG controller will take over to regulate the PWM duty {cycles} within these compromised DERs according to the voltages of legitimate DERs. The overarching goal is to reduce the power sharing imbalance among all DERs, while the accurate power sharing is hard due to the uncertainties of local ZIP loads within DERs.} In experimental studies, the performance of holding and adaptive strategies that regulate the PWM duty {cycles} of affected DERs are demonstrated and compared.

{\color{black}\noindent\textbf{Remark 1:} Although the present hierarchical detection and mitigation framework is developed for networked DC microgrids, the basic idea of hierarchical detection and mitigation can be extended to AC microgrids and hybrid AC/DC microgrids. Such an extension requires a dedicated small-signal system model for AC microgrids in the $dq$ frame, with the impact of the $abc/dq$ coordinate transformation explicitly incorporated into the spoofing attack propagation and closed-loop stability analysis \cite{sahoo2021ac,yan2019smallsignal}. In particular, different inverter control methods, such as grid-following and grid-forming, can be integrated into the small-signal dynamics, but their synchronization mechanisms, control-loop interactions, and stability contributions should be modelled in a coordinated manner \cite{9796617,li2022duality}. For hybrid AC/DC microgrids, the handling of the DC-AC interface requires explicit modelling of the interlinking converter together with the coupled AC-side and DC-side dynamics. Existing studies show that unified linearised state-space models can be derived for specific hybrid AC/DC microgrid architectures when the AC subgrid, the DC subgrid, the interlinking converter, and the associated communication or control loops are modelled jointly; however, the resulting formulation remains topology-dependent and control-dependent rather than universally transferable across all hybrid configurations \cite{liu2011hybrid,chang2021interlinking,yoo2020hybrid,espina2021hybrid}. Moreover, the global monitoring metrics at the microgrid layer need specific redesign to timely capture adversarial impacts on voltage, frequency, phase-angle, power-sharing, and interlinking-converter dynamics, while the perturbation of electrical parameters should additionally account for the transient-stability constraints of low-inertia renewable-dominated power systems \cite{he2022lowinertia}. When the MG-side electrical topology changes from a radial structure to a meshed one, the MG-layer reconstruction in \eqref{eq: voltage reconstruction} should also be reformulated according to the corresponding network nodal equations, such that the voltage-bias information required by the DER-layer recursive mitigation in \eqref{eq: recursive reconstruction form} can still be provided. This topology-dependent reformulation, together with the above AC and hybrid AC/DC extensions, is beyond the scope of this paper and will be investigated in our future work. It is also noted that the scalability of the proposed method is strong because the DER-layer implementation mainly requires local linear observers and recursive estimators, without involving complex data-processing modules or high-dimensional learning models.}

{\color{black}\noindent\textbf{Remark 2:} {The proactive detection mechanism in this work relies on perturbing circuit-side electrical parameters to create information asymmetry against attack biases that are aligned with the pre-perturbation circuit dynamics. Therefore, purely software-side retuning, such as modifying control gains {\cite{mokhtar2019adaptiveDroop}} or virtual-impedance parameters {\cite{wu2017virtualImpedance}} only, is not sufficient because it does not directly alter the physical consistency relation between voltage and current measurements. From an implementation viewpoint, this perturbation can be retrofitted to an already deployed converter by adding a compact impedance-selection module at the converter output stage, where calibrated resistance branches are switched after receiving the perturbation signal from the MG controller. Such a hardware pathway is consistent with the reconfigurable and auxiliary-circuit design philosophy, where additional reconfigurable branches or plug-play auxiliary cells are integrated into converter hardware without redesigning the entire power stage \cite{abramson2018reconfigurable}. As a concrete low-voltage example, a 504W commercial converter such as the MEAN WELL SD-500L-48 is listed at \$157.7, while a retrofit branch based on one 20A DC contactor and one 0.1$\Omega$/50W chassis resistor requires about \$24.8 of additional hardware cost at listed distributor prices accessed on April 24, 2026 \cite{digikey_sd500l48,digikey_aev20eb,digikey_hs50r1j}. Hence, the proposed method does not require redesigning the original converter control stack, but can be integrated through a compact add-on hardware module with bounded cost and limited installation effort.}}

\begin{figure}[!h]
    \centering
    \includegraphics[width=1\linewidth]{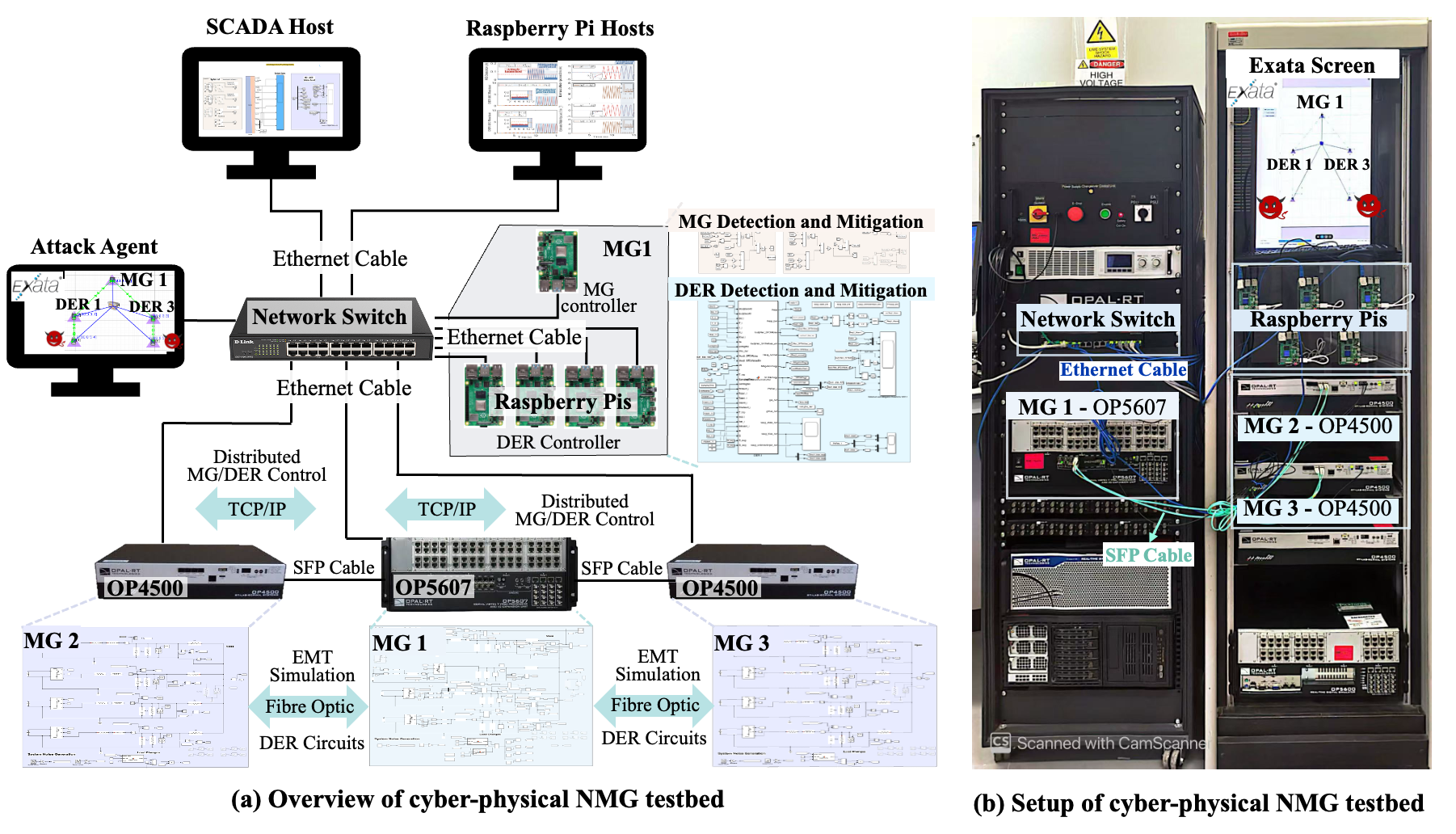}
    \caption{\color{black}This figure presents the cyber-physical NMG testbed used for validation, consisting of one OPAL-RT OP5607 and two OP4500 real-time simulators, Keysight EXata communication emulator, and Raspberry Pis as co-simulation interface.}
    \label{fig:experimentaltestbed}
\end{figure}

\section{Experimental Results}
In this section, we conduct experiments in a {cyber-physical} NMG testbed that integrates real-time power and communication simulators, as shown in Fig. \ref{fig:experimentaltestbed}, to validate the performance of proposed hierarchical defence framework. The cyber and physical topologies are shown in Fig. \ref{fig:system-diagram}. In particular, the EMT dynamics of three MGs are simulated in one OPAL-RT OP5607 and two OP4500 real-time simulators, where the interconnection between MGs is modelled as equivalent circuits with data transmitted via {fibre optic (SFP) cables}. The distributed MG/DER control signals are interacted through an Ethernet TCP/IP communication network. In MG 1, the MG and DER controllers of DERs 1 and 3 are implemented in three Raspberry Pis to test the performance of proposed hierarchical defence framework. Moreover, two extra Raspberry Pis are used as data transfer stations, which, together with existing three Raspberry Pis, are mapped into Exata as five virtual nodes. Then, the cyber library of EXata can be applied to implementing realistic data modification attacks against the mapped sensor links, mimicking the data manipulation effect of sensor spoofing attack \eqref{eq: spoofing attack V}, \eqref{eq: spoofing attack I}. The performance of proposed hierarchical framework is shown under two scenarios according to if attack vectors are generated from \eqref{eq: stealthy attack vector}.

\begin{figure}[!h]
    \centering
    \includegraphics[width=1\linewidth]{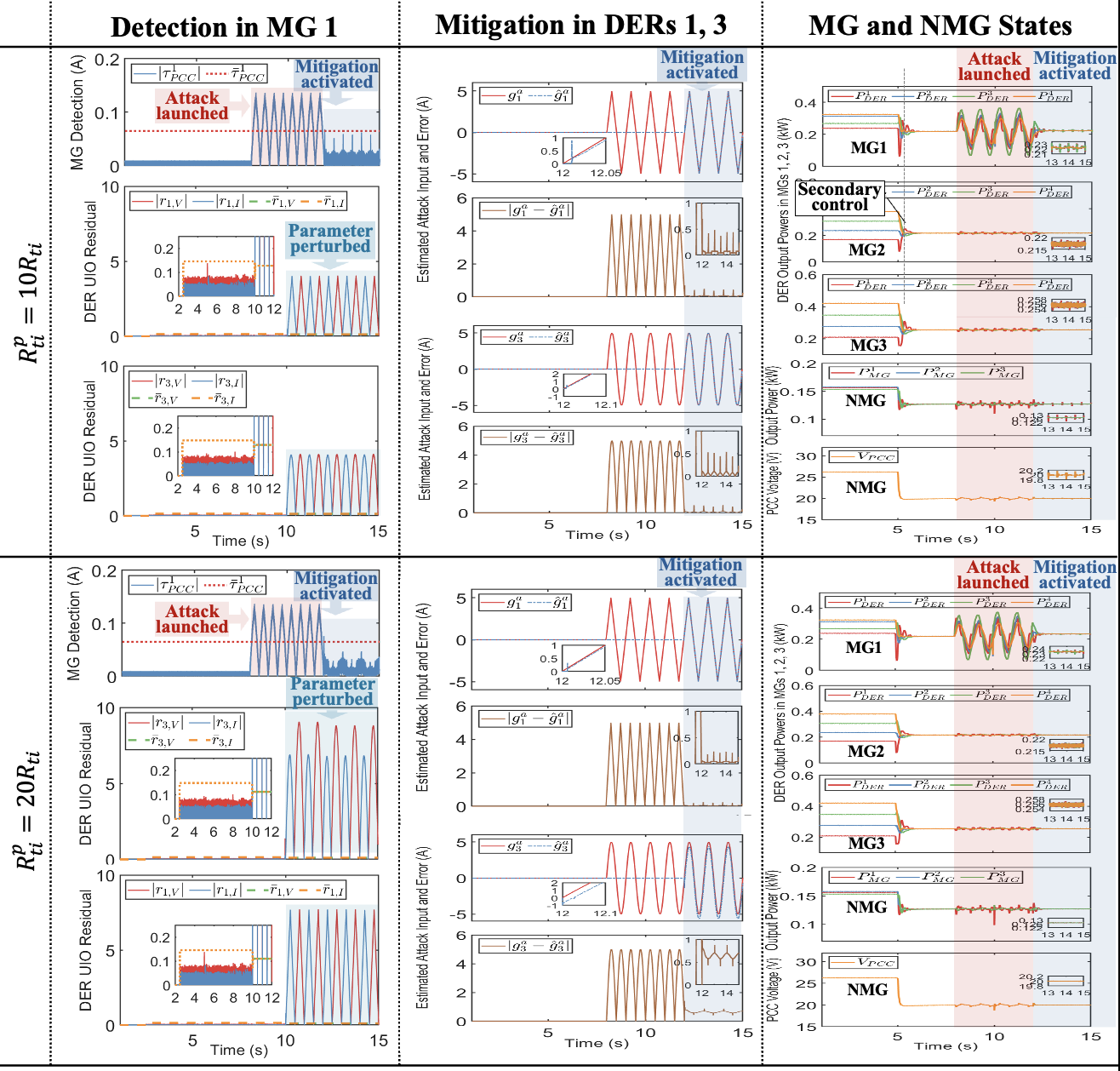}
    \caption{This figure shows the performance of proposed hierarchical detection and mitigation framework under intelligent sensor spoofing attacks (in DERs $1$ and $3$) satisfying \eqref{eq: stealthy attack vector}.
    % varying perturbation strengths. 
    In particular, two perturbation strengths on resistance $R_{ti}$ (i.e., $10R_{ti}$ and $20R_{ti}$) are selected, corresponding to two rows from top to bottom. In each case, the detection and mitigation metrics in both MG and DER layers are showcased in the first two columns, and the system states under the hierarchical framework are illustrated in the last column. The timeline of main events is: 1) $t=5$s, secondary control activated; 2) $t=8$s, attack launched; 3) $t=10$s, perturbation activated; 4) $t=12s$, mitigation activated, which are highlighted with different background colors.
    }
    \label{fig:CaseI_StealthyAttackVector}
\end{figure}

{\color{black}To clarify the selection of the MG-layer detection threshold, Table \ref{tab:falsealarm_threshold} reports the false alarm rates under different electrical-parameter uncertainty ratios and different thresholds $\bar{\tau}_{PCC}^m$. Here, the uncertainty ratio represents the normal mismatch of electrical parameters, such as the equivalent parameter $1/G_{DER_j}^{PCC_m}$ used in the MG-layer detector \eqref{eq: current flow based verification}. In the false-alarm analysis, this electrical-parameter uncertainty is generated according to a normal distribution, and the listed uncertainty ratios specify the corresponding standard-deviation levels relative to the nominal electrical parameter. A $5\%$ electrical-parameter variation is used to represent practical electrical-parameter drift induced by environment-temperature changes or aging-related resistance mismatch \cite{liu2019droop, qin2025adaptive}. System noises with bounds $\bar{\bm{\omega}}_i=\bar{\bm{\rho}}_i=0.01[1,1]^{\rm T}$ are integrated into the experiments. The last column gives the maximal power-sharing error that can remain below the corresponding threshold, which represents the maximal missed attack impact from the perspective of power-sharing imbalance. The results show a clear trade-off: a smaller threshold improves detection sensitivity and reduces the maximal missed attack, but it also significantly increases the false alarm rate when the electrical-parameter uncertainty ratio becomes larger. For example, when $\bar{\tau}_{PCC}^m=0.02$, the maximal missed power-sharing error is only $2$W, but the false alarm rate increases from $10.02\%$ to $70.97\%$ as the electrical-parameter uncertainty ratio increases from $0.01$ to $0.05$. In contrast, $\bar{\tau}_{PCC}^m=0.2$ almost eliminates false alarms across all uncertainty ratios, but allows up to $9$W missed power-sharing error. Therefore, the threshold $\bar{\tau}_{PCC}^m=0.065$ used in the experiment provides a compromise between robustness and sensitivity, keeping the false alarm rate low under mild electrical-parameter uncertainty while limiting the missed power-sharing error to $3.5$W.}

\begin{table}[!h]
    \centering
    \caption{{\color{black}False alarm rate and maximal missed attack impact under different MG-layer thresholds and electrical-parameter uncertainty ratios.}}
    \label{tab:falsealarm_threshold}
    {\color{black}
    \footnotesize
    \setlength{\tabcolsep}{2pt}
    \renewcommand{\arraystretch}{1.12}
    \resizebox{\linewidth}{!}{%
    \begin{tabular}{c c c c c c c}
        \toprule[1pt]
        \textbf{$\bar{\tau}_{PCC}^m$} & \multicolumn{5}{c}{\textbf{\shortstack[c]{False alarm rate under uncertainty ratios of $1/G_{DER_j}^{PCC_m}$}}} & \multirow{2}{*}{\textbf{\shortstack[c]{\\~Maximal missed attack impact\\on power-sharing error (W)}}} \\
        \cmidrule(lr){2-6}
        & \textbf{$0.01$} & \textbf{$0.02$} & \textbf{$0.03$} & \textbf{$0.04$} & \textbf{$0.05$} & \\
        \midrule
        $0.02$  & $10.02\%$ & $36.54\%$ & $53.80\%$ & $64.20\%$ & $70.97\%$ & $2$ \\
        $0.065$ & $6.25{\times}10^{-4}\%$ & $0.34\%$ & $4.66\%$ & $13.16\%$ & $22.60\%$ & $3.5$ \\
        $0.1$   & $0$ & $1.875{\times}10^{-3}\%$ & $0.23\%$ & $2.10\%$ & $6.36\%$ & $5$ \\
        $0.2$   & $0$ & $0$ & $0$ & $1.875{\times}10^{-3}\%$ & $2.5{\times}10^{-3}\%$ & $9$ \\
        \bottomrule[1pt]
    \end{tabular}}}
\end{table}

\subsection{\texorpdfstring{Proactive Detection and Recursive Impact Mitigation under Intelligent Sensor Spoofing Attacks Satisfying \eqref{eq: stealthy attack vector}}{Proactive Detection and Recursive Impact Mitigation under Intelligent Sensor Spoofing Attacks Satisfying the Stealthy Attack Vector}}\label{sub:scenarioI}

The intelligent sensor spoofing attack satisfying \eqref{eq: stealthy attack vector} can bypass the original UIO-based detector, and thus requires to perturb resistance $R_{ti}$ for effective detection and mitigation.
As shown in Fig. \ref{fig:CaseI_StealthyAttackVector}, the performance of proposed hierarchical framework under two perturbation strengths on resistance $R_{ti}$ (i.e., $10R_{ti}$ and $20R_{ti}$) are demonstrated.
% The attack detection and impact mitigation performance of proposed hierarchical framework are shown in Fig. \ref{fig:CaseI_StealthyAttackVector}, where two perturbation strengths on resistance $R_{ti}$ (i.e., $10R_{ti}$ and $20R_{ti}$) are introduced to showcase their impacts on the defence performance.
% The three perturbations include $R_{ti}^p = 3R_{ti}$, $10R_{ti}$, $20R_{ti}$, where $R_{ti}^p$ denotes the DER circuit resistance after perturbation. The results are divided into detection, mitigation, and system states to clearly show the performance from different angles. 
The event timeline is: 1) At $t=5$s, the secondary control at both DER and MG layers are activated; 2) At $t=8s$, two sensor spoofing attacks satisfying \eqref{eq: stealthy attack vector} are launched against DERs 1 and 3 with triangle and sine fake unknown inputs ($g_i^a, i \in \{1,3\}$), respectively; 3) At $t=10$s, resistance perturbations are applied to all DERs within MG 1; 4) At $t=12$s, impact mitigation scheme \eqref{eq: impact mitigation} is enabled.
% estimated bias vector $\widehat{\bm{\psi}}_i$ is subtracted from 
% DER $3$ is tripped as it has larger rated output power than DER $1$. 
The sampling period $T_{samp} = {4\times10^{-4}}$s. The subsequent contents follow to explain the observed {phenomena}.

\subsubsection{Proactive Attack Detection}
The detection results include the anomaly metric $\tau_{PCC}^1$ based on the Kirchhoff current law in MG 1 and UIO residuals $\bm{r}_1$, $\bm{r}_3$. 
% Here, the detection threshold for $\tau_{PCC}^1$ is set as $\bar{\tau}_{PCC}^1 = 0.065$A to tolerate the measurement noise. 
According to the first column of Fig. \ref{fig:CaseI_StealthyAttackVector}, both $\tau_{PCC}^1$ and $\bm{r}_i$ are bounded by corresponding thresholds until the launch of spoofing attacks at $t=8$s, indicating that the applied detection metrics are insensitive to the activation of secondary control. During $t\in[8,10]$s, the launched attacks immediately trigger the MG detection alarm \eqref{eq: condition breach} but can {bypass} the UIO-based detector \eqref{eq: UIO model} as the applied attack dynamics \eqref{eq: stealthy attack vector} align with the DER dynamics \eqref{eq: Discrete DER SS Model}. When the resistance perturbations are applied, the affected DERs 1 and 3 can be quickly identified as the related UIO residuals exceed the thresholds swiftly, validating the improved detection capability by perturbing $R_{ti}$. The {increase of resistance perturbation from $10R_{ti}$ to $20R_{ti}$ amplifies} the detection residuals under attacks, which further {improves} the mitigation performance as seen on $\tau_{PCC}^1$ after the estimated bias vector $\widehat{\bm{\psi}}_i$ is adopted.

\subsubsection{Recursive Impact Mitigation} The impact mitigation performance is shown via the estimated fake unknown input $g_i^a, i \in \{1,3\}$ in the second column of Fig. \ref{fig:CaseI_StealthyAttackVector}. Since the increase of resistance perturbation will amplify detection residuals, a relatively smaller bias injection can exceed the detection threshold, thereby triggering the bias estimation. In such case, the spikes on the reconstruction errors resulting from insufficient bias injections can be significantly decreased accordingly. Another phenomenon is that, for DER $3$, a steady-state estimation error emerges after $20R_{ti}$ perturbation, while no steady-state estimation error is observed in DER $1$. The difference between DERs 1 and 3 attributes to the variation of local ZIP loads, which may make the steady state before perturbation cannot be explained the DER dynamics with perturbed $R_{ti}$.
% The resistance perturbation strength will not affect the reconstruction of voltage bias, which, akin to the MG detection, is directly related to the measurement noise \eqref{eq: voltage reconstruction}. On the contrary, the increasing perturbation strength would lead to a improved reconstruction of current bias. Specifically, larger resistance perturbations can yield a smaller steady-state reconstruction error, as the increased UIO residual more accurately reflects the bias information \eqref{eq: recursive reconstruction form}. 
% Moreover, a larger UIO residual allows for a smaller minimal bias injection to exceed the detection threshold, thereby triggering bias reconstruction. Consequently, the spikes in reconstruction error become more negligible.

\subsubsection{System States}
The system states, consisting of the DER output powers, MG output powers, and PCC voltage, under the proposed hierarchical framework are shown in the third column of Fig. \ref{fig:CaseI_StealthyAttackVector}. When looking into the details within MGs and of NMG, the attack impact can be largely mitigated after activating the impact mitigation scheme \eqref{eq: impact mitigation}.
In particular, the larger resistance perturbation on $R_{ti}$ can lead to smoother system states after impact mitigation because of the improved estimation accuracy. Moreover, the larger perturbation will induce more intense transient state fluctuations but can quickly converge to expected trajectories owing to the stable eigenvalues \eqref{eq: eigenvalues}.
At the same time, the increased resistance $R_{ti}$ will slightly improve the steady-state output power of DER 1 within MG 1 under the same ZIP loads. {This improvement should be interpreted together with the additional power loss caused by resistance perturbation, which motivates future optimisation of the perturbation ratio rather than treating larger resistance as universally preferable.}
% As shown in Fig. \ref{fig:powerloss} (an independent study without involving attacks), the increase of perturbation ratio $R_{ti}^p/R_{ti}$ leads to linearly increasing power line loss, which is calculated as the variation of power consumption on $R_{ti}$ before and after perturbation, denoted by $\Delta P_{MG}^{1,loss} = \sum_{i=1}^4\big((I_{DER}^{ti,ssp})^2R_{ti}^p - (I_{DER}^{ti,ss})^2R_{ti}\big)$, with $I_{DER}^{ti,ss}$ and $I_{DER}^{ti,ssp}$ denoting the steady-state output currents of DER $i$ in MG 1 before and after perturbation. 
{In a nutshell, the result exemplifies that resistance perturbation can successfully identify the compromised DERs and largely mitigate the attack impacts in cyber-physical NMGs, while the associated power-loss cost should be balanced against the detection and mitigation improvement.}

{\color{black}To further verify the feasibility of the proposed framework beyond the buck-converter cases, DER 2 in MG 1 is reconfigured as a boost converter and the same defence strategy is deployed on it, while DERs 1 and 3 remain compromised buck converters. Different from the LTI model in \eqref{eq: DER SS Model}, the averaged dynamics of boost-converter DER 2 become a linear time-varying state-space system because the duty cycle $d_i$ varies with the operating status, i.e.,
\begin{align}
% \dot{\bm{x}}_i &= A_i(d_i)\bm{x}_i + B_i u_i + E_i g_i + \bm{\omega}_i, \quad \bm{y}_i = \bm{x}_i + \bm{\rho}_i, \nonumber \\
A_i(d_i) &= \begin{bmatrix}
-\frac{1}{R_{Li}C_{ti}} & \frac{1-d_i}{C_{ti}} \\
-\frac{1-d_i}{L_{ti}} & -\frac{R_{ti}}{L_{ti}}
\end{bmatrix}.
\end{align}
As a result, the generated bias waveforms are no longer strictly sinusoidal because the duty-cycle-dependent boost-converter dynamics introduce extra nonlinearity into the closed-loop behaviour (seen in Fig. \ref{fig:boostconverterloadvariation}).}
{\color{black}To handle this time-varying nonlinearity, the instantaneous duty cycle is treated as a known scheduling variable, based on which the discrete boost-converter model is updated online at each sampling instant. Accordingly, the UIO-related matrices are recalculated with the current duty cycle, such that the observer parameters remain matched to the converter operating point rather than being fixed at a single linearisation condition. In the same manner, the Kalman filter parameters for bias estimation are synchronously refreshed from the duty-cycle-dependent state-space model, enabling the estimator to track the time-varying residual-bias relation more accurately. The proposed framework can be applied to buck-boost converters in a similar way by formulating their averaged converter dynamics as duty-cycle-dependent state-space models and then constructing the UIO and Kalman-filter-based bias estimator accordingly \cite{tan2015dc}.}

\begin{figure}[!h]
    \centering
    \includegraphics[width=0.5\linewidth]{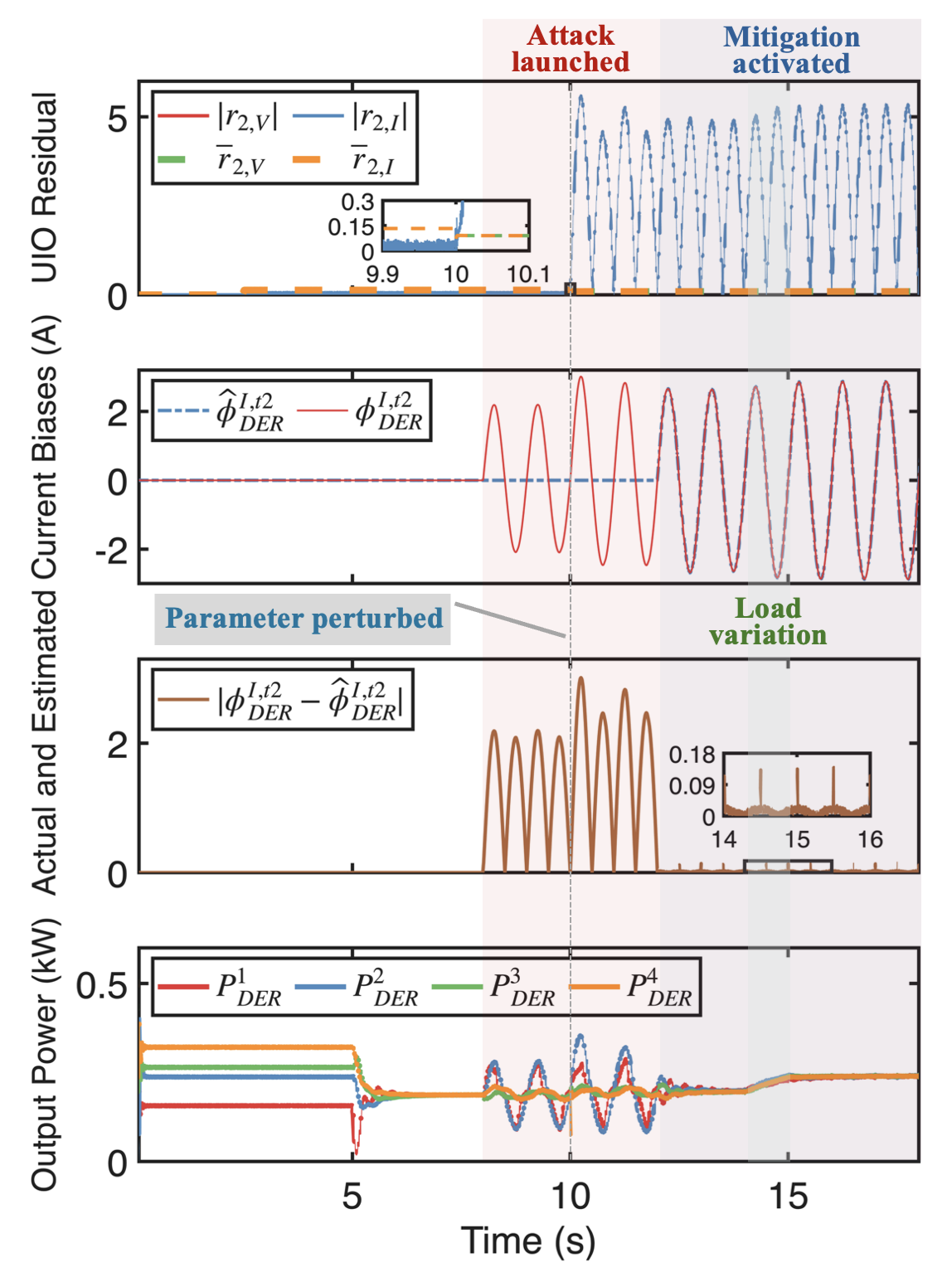}
    \caption{{\color{black}This figure shows the effectiveness of the proposed defence framework when DER 2 in MG 1 is configured as a boost converter. 
    % Three DERs 1, 2, and 3 are compromised, where DERs 1 and 3 are buck converters and DER 2 is the boost converter. The stealthy attack is perceived by the UIO residual after the defence is activated at $t=10$s, and the injected current bias in the boost converter is accurately estimated by the Kalman filter despite the duty-cycle-dependent converter dynamics. During $t\in[14,15]$s, the load currents of all DERs increase linearly by $25\%$, showing that the proposed method remains compatible with load variations treated as unknown inputs. The voltage-bias estimation behaves similarly and is omitted for brevity.
    }}
    \label{fig:boostconverterloadvariation}
\end{figure}

{\color{black}The corresponding results are shown in Fig. \ref{fig:boostconverterloadvariation}. At $t=10$s, the UIO detection residual clearly perceives the stealthy attack and enables the Kalman-filter-based estimator to accurately reconstruct the injected current bias even in the presence of the boost-converter-induced time-varying dynamics. This indicates that updating the UIO and Kalman filter parameters with the duty cycle can effectively accommodate the nonlinearity induced by boost-converter operation. Moreover, the steady-state current-bias estimation error is smaller than $0.18$A, which {substantially reduces} the attack impact on the system states. During the load variation interval $t\in(14,15)$s, all DER load currents are increased linearly by $25\%$. The result shows that this operating change can still be well handled by the UIO because it is regarded as an unknown input, and the proposed defence remains compatible with such load variations without losing attack-detection and mitigation capability.}

\begin{figure}[!t]
    \centering
    \includegraphics[width=0.85\linewidth]{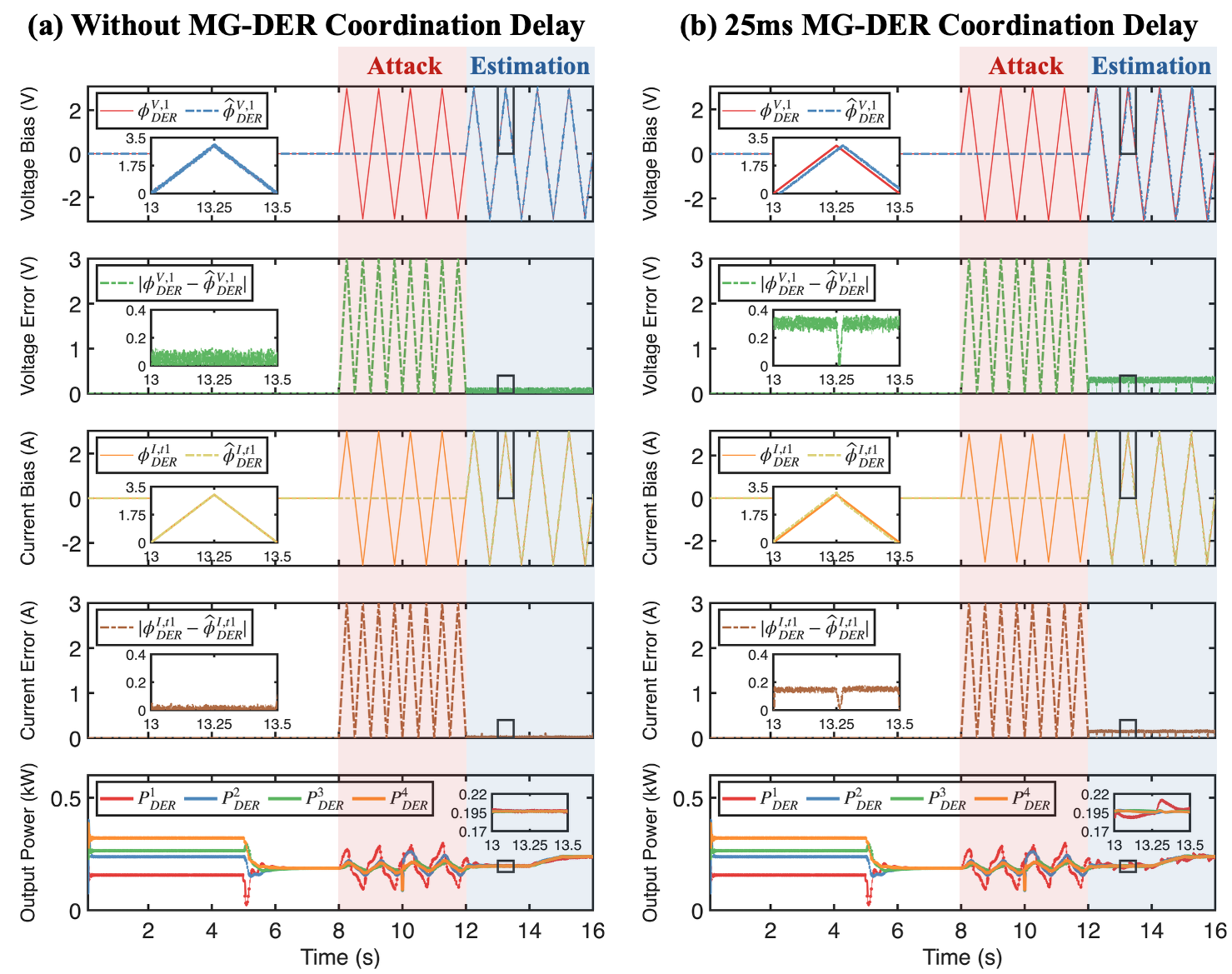}
    \caption{{\color{black}This figure shows the performance of the proposed hierarchical detection and mitigation framework under the sensor spoofing attack not satisfying \eqref{eq: stealthy attack vector}. 
    % Two MG-DER coordination settings are compared: ideal coordination without delay and coordination with a $25$ms delay. In each column, the first and second subplots show the voltage bias reconstruction and its absolute estimation error, the third and fourth subplots show the current bias reconstruction and its absolute estimation error, and the last subplot shows the output powers of all DERs in MG 1 under the mitigation scheme.
    }}
    \label{fig:CaseII_GeneralAttackVector}
\end{figure}

\subsection{\texorpdfstring{Impact Mitigation under Sensor Spoofing Attacks not Satisfying \eqref{eq: stealthy attack vector}}{Impact Mitigation under Sensor Spoofing Attacks not Satisfying the Stealthy Attack Vector}} 
When the sensor spoofing attack do not satisfy \eqref{eq: stealthy attack vector}, it can be perceived by the UIO-based detector \eqref{eq: UIO model} without resistance perturbation. The mitigation strategies differ depending on the number of affected DERs and their performance will be validated via two cases.

\begin{figure}[!t]
    \centering
    \includegraphics[width=0.65\linewidth]{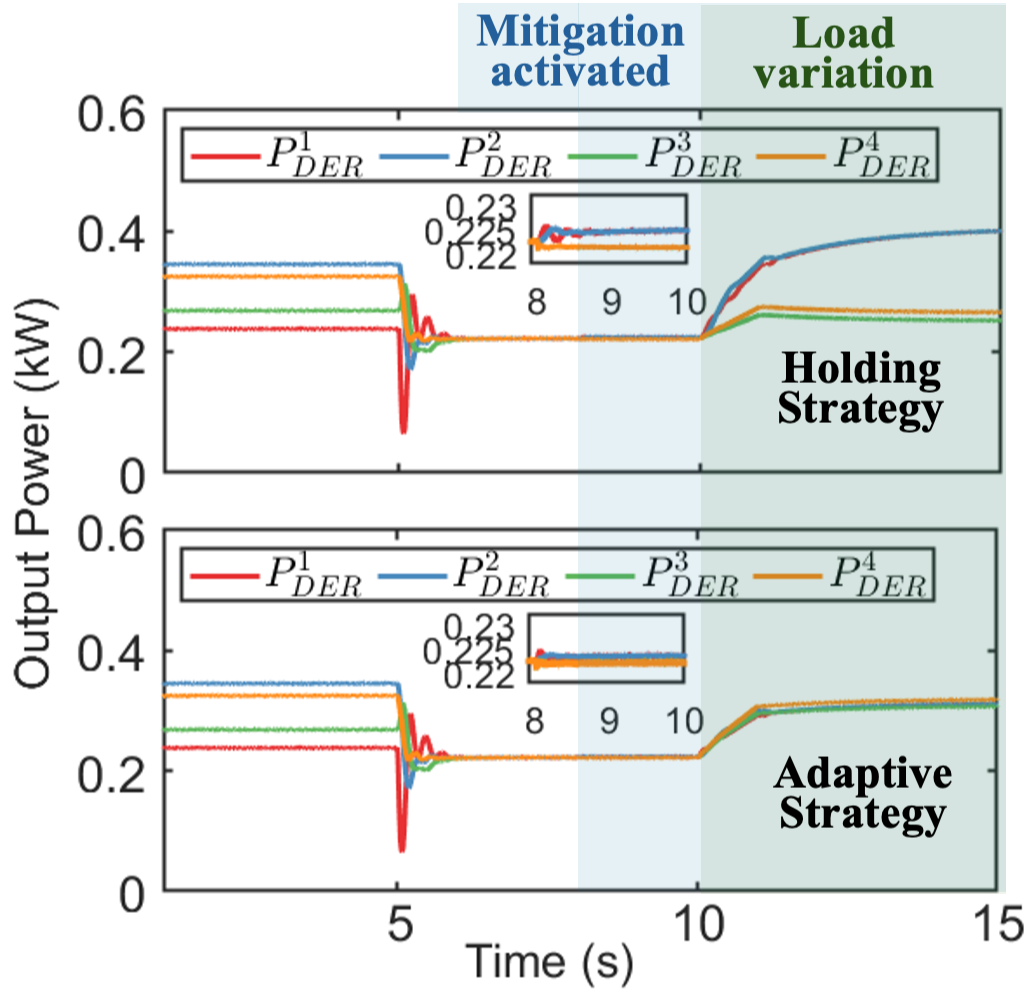}
    \caption{This figure shows the performance of holding and adaptive mitigation strategies when DERs 1 and 3 are under the sensor spoofing attacks not satisfying \eqref{eq: stealthy attack vector}. 
    % The attack is launched at $t=8$s and the mitigation is activated immediately after detecting anomaly. During $t\in[10,11]$s, the linearly increasing load variations are introduced to the local and PCC ZIP loads.
    }
    \label{fig:caseii_morethan1DER}
\end{figure}
\subsubsection{\texorpdfstring{Recursive impact mitigation with $1$ affected DER}{Recursive impact mitigation with one affected DER}}
If only DER $1$ within MG $1$ is affected, then its legitimate voltage can be estimated in the MG layer via \eqref{eq: voltage reconstruction}. Based on that, its normal current can be recursively estimated in the DER layer through \eqref{eq: recursive reconstruction form}. {\color{black}Therefore, the bias estimation and impact mitigation performance are illustrated in Fig. \ref{fig:CaseII_GeneralAttackVector}, where ideal MG-DER coordination and $25$ms delayed MG-DER coordination are compared. The event timeline here is the same as that in the previous subsection. Under ideal coordination, the post-mitigation voltage- and current-bias estimation errors over $t\in[12,16]$s are bounded by approximately $0.124$V and $0.119$A, with mean errors of $0.044$V and $0.011$A, respectively. When a $25$ms MG-DER coordination delay is introduced, the corresponding maximum errors increase to approximately $0.382$V and $0.176$A, and the mean errors increase to $0.295$V and $0.143$A, respectively. The output-power responses also show a larger average DER power-spread after mitigation, increasing from about $0.0048$kW under ideal coordination to $0.0090$kW with the $25$ms delay. These results indicate that MG-DER coordination delay mainly degrades the recursive bias reconstruction accuracy and slightly worsens the post-mitigation power-sharing behaviour, while the proposed mitigation still keeps the system states bounded and substantially improved compared with the attacked interval before mitigation. It is noted that the investigation of delay-tolerant mitigation schemes that explicitly compensate MG-DER coordination latency will be left as future work.}
\subsubsection{\texorpdfstring{Open-loop central regulation with more than $1$ affected DERs}{Open-loop central regulation with more than one affected DER}}
This subsection considers that DERs $3$ and $4$ are affected by the sensor spoofing attacks not satisfying \eqref{eq: stealthy attack vector}. These attacks can be effectively detected by the UIO-based detector \eqref{eq: UIO model}, but their recursive estimation faces challenges due to insufficient observability.
{\color{black}The rationale is that, when multiple DERs are compromised and no further implication on the bias-injection design is available, fully {removing} all attack impacts through recursive bias estimation becomes challenging. Instead of pursuing exact bias compensation under insufficient observability, the central idea is to fall back from the local closed-loop DER control to open-loop central MG control. This fallback ceases the propagation of attack impacts at the compromised DERs by preventing falsified local measurements from continuously entering their closed-loop controllers, at the expense of sacrificing part of the global control performance such as power sharing.} Therefore, the demonstrated mitigation first disables the local closed-loop control within affected DERs immediately after perceiving anomalies at time instant $k_a$, and implements the open-loop central regulation of PWM duty {cycles} within these DERs. The adopted strategies include
% \small
\begin{align}
    \text{Holding:}\ \ d_{DER}^i (k) &= d_{DER}^i (k-1) \label{eq: holding strategy} \\
    \text{Adaptive:}\ \ d_{DER}^i (k) &= \big(1 + \kappa_i(k)\big) d_{DER}^i (k) \label{eq: adaptive strategy}
\end{align}where $\kappa_i(k) = \sum_{j\in\mathcal{N}_{DER}^{MG_1}/\mathcal{N}_{DER_a}^{MG_1}} \frac{V_{DER}^j(k) - V_{DER}^j(k-1)}{\big|\mathcal{N}_{DER_a}^{MG_1}\big|V_{DER}^j(k-1)}$ calculates the regulation ratio on PWM duty {cycle} according to the voltage variation trend of legitimate DERs. As shown in Fig. \ref{fig:caseii_morethan1DER}, the holding strategy \eqref{eq: holding strategy} can prevent the attack impact from getting worse and propagating to other normal DERs, while showing limitations in {pursuing} power sharing among DERs especially under load variation. On the other hand, the adaptive strategy \eqref{eq: adaptive strategy} can track the voltage variation trend of legitimate DERs within MG $1$, therefore decreasing the power sharing error. However, the transient and steady-state performance of intelligent regulation strategies such as \eqref{eq: adaptive strategy} especially when considering heterogeneous converter dynamics {requires} further analytical investigation and {is} left for future {work}.

% In practice, the computation periods in DER and MG layers could have discrimination,different from the previous cases assuming that the MG and DER mitigation periods are the same $T_{MG}^{miti} = T_{DER}^{miti} = 200\mu$s. Therefore, the impact of asynchronous cross-layer mitigation is investigated by varying $T_{MG}^{miti}$ from $10$ms to $1000$ms, which align with realistic scenarios where the communication delay is not negligible. According to Fig. \ref{fig:asynchronousmitigation}, the increase of $T_{MG}^{miti}$ significantly decreases the accuracy of reconstructed voltage bias as the update frequency severely lags behind the bias dynamics \eqref{eq: stealthy attack vector}. However, the accuracy of reconstructed current bias is surprisingly good despite the inaccurate voltage bias information. This phenomenon occurs because the actual current bias is over ten times greater than the actual voltage bias, enabling the reconstructed current bias to maintain high accuracy even when assuming zero voltage bias and attributing the entire UIO residual to the current bias. Therefore, the attack impact can be largely mitigated even with the worst $T_{MG}^{miti} = 1$s. This study verifies the practical value of proposed framework, showcasing its satisfactory performance despite non-negligible communication delays between DERs and MGs.

\begin{figure}[!h]
    \centering
    \includegraphics[width=0.85\linewidth]{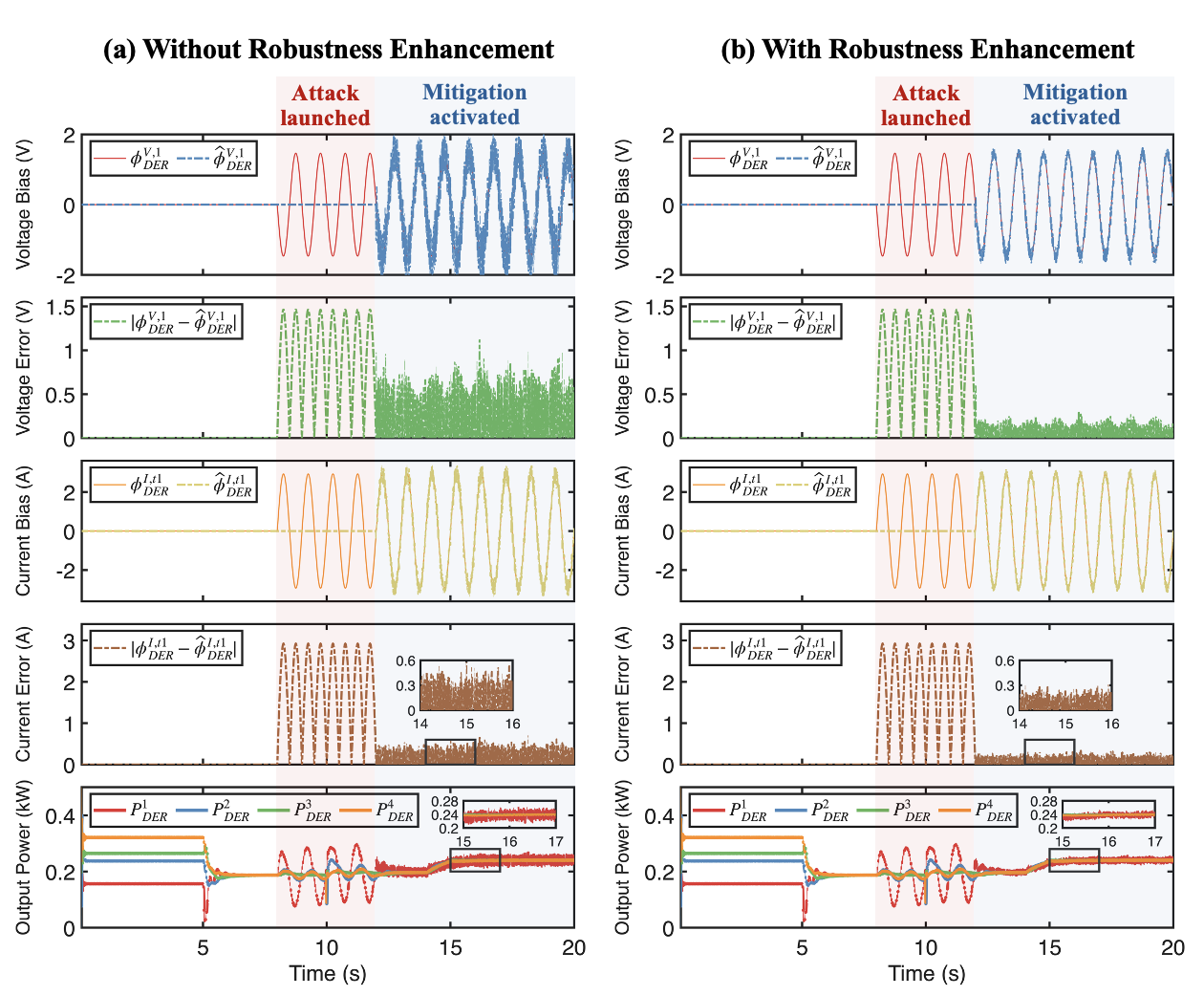}
    \caption{{\color{black}This figure shows the robustness of the proposed detection and mitigation framework under electrical-parameter variations, where the electrical parameters $R_{t1}$, $C_{t1}$, $L_{t1}$ and the equivalent ZIP load $Z_{DER}^{ZIP,1}$ are perturbed by approximately $10\%$ normal-distribution variations in the residual calculation and bias estimation. 
    % The proposed method remains stable under these uncertainties, while nontrivial voltage- and current-bias estimation errors are induced by the parameter mismatch. After the initial transient, the robust UIO design keeps the post-transient voltage- and current-bias estimation errors within approximately $0.590$V and $0.340$A, respectively, and improves the power-output response in MG 1.
    }}
    \label{fig:robustnesstoparametersensitivity}
\end{figure}

{\color{black}\subsection{Sensitivity to Parameter Uncertainties and Enhancement}
The above results are obtained under the nominal electrical parameters used in the residual calculation and bias estimation. In practice, however, electrical parameters may deviate from their nominal values due to temperature, aging, load changes, and modelling mismatch \cite{liu2019droop, qin2025adaptive}. To evaluate this impact, the electrical parameters $R_{t1}$, $C_{t1}$, $L_{t1}$ and the equivalent ZIP load $Z_{DER}^{ZIP,1}$ are injected with approximately $10\%$ variations generated from a normal distribution when calculating the UIO residuals and estimating the injected biases. As shown in Fig. \ref{fig:robustnesstoparametersensitivity}, the proposed framework remains stable under these electrical-parameter uncertainties and the attack impacts can still be mitigated. Nevertheless, the parameter mismatch introduces nontrivial bias-estimation errors, especially in the current-bias channel, because the residual-to-bias reconstruction relation is directly affected by the uncertain converter dynamics.}

{\color{black}To enhance robustness against such parameter sensitivities, the UIO parameters are redesigned by following the robust UIO philosophy in \cite{gao2016robustuio}, where the unknown inputs that can be structurally decoupled are eliminated by the UIO condition, while the remaining disturbance components caused by parameter mismatch are attenuated through an optimisation criterion. Specifically, for each uncertainty vertex $\vartheta\in\mathcal{V}$ generated from the normal-distribution range of electrical parameters, the discrete DER model $(A_{di}(\vartheta),B_{di}(\vartheta),E_{di}(\vartheta))$ is reconstructed and the UIO matrices are parameterised as
\begin{align}
    T_{di}(\vartheta) &= \begin{bmatrix} g_1 \bm{n}_{di}^{\rm T}(\vartheta) \\ g_2 \bm{n}_{di}^{\rm T}(\vartheta) \end{bmatrix}, \quad
    \bm{n}_{di}^{\rm T}(\vartheta)E_{di}(\vartheta) = 0, \nonumber\\
    H_{di}(\vartheta) &= I - T_{di}(\vartheta), \nonumber\\
    \hat{K}_{di}(\vartheta) &= T_{di}(\vartheta)A_{di}(\vartheta) - F_{di}T_{di}(\vartheta),
\end{align}
where $\bm{n}_{di}(\vartheta)$ spans the left nullspace of the unknown-input vector $E_{di}(\vartheta)$. This construction guarantees $T_{di}(\vartheta)E_{di}(\vartheta)=0$ for the decouplable unknown input under every sampled uncertainty vertex. The remaining free variables are $\theta=\{g_1,g_2,F_{di}\}$, where $F_{di}$ is allowed to be a full $2\times2$ Schur matrix instead of being restricted to the diagonal form used in the nominal UIO design.}

\begin{figure}[!h]
    \centering
    \includegraphics[width=0.85\linewidth]{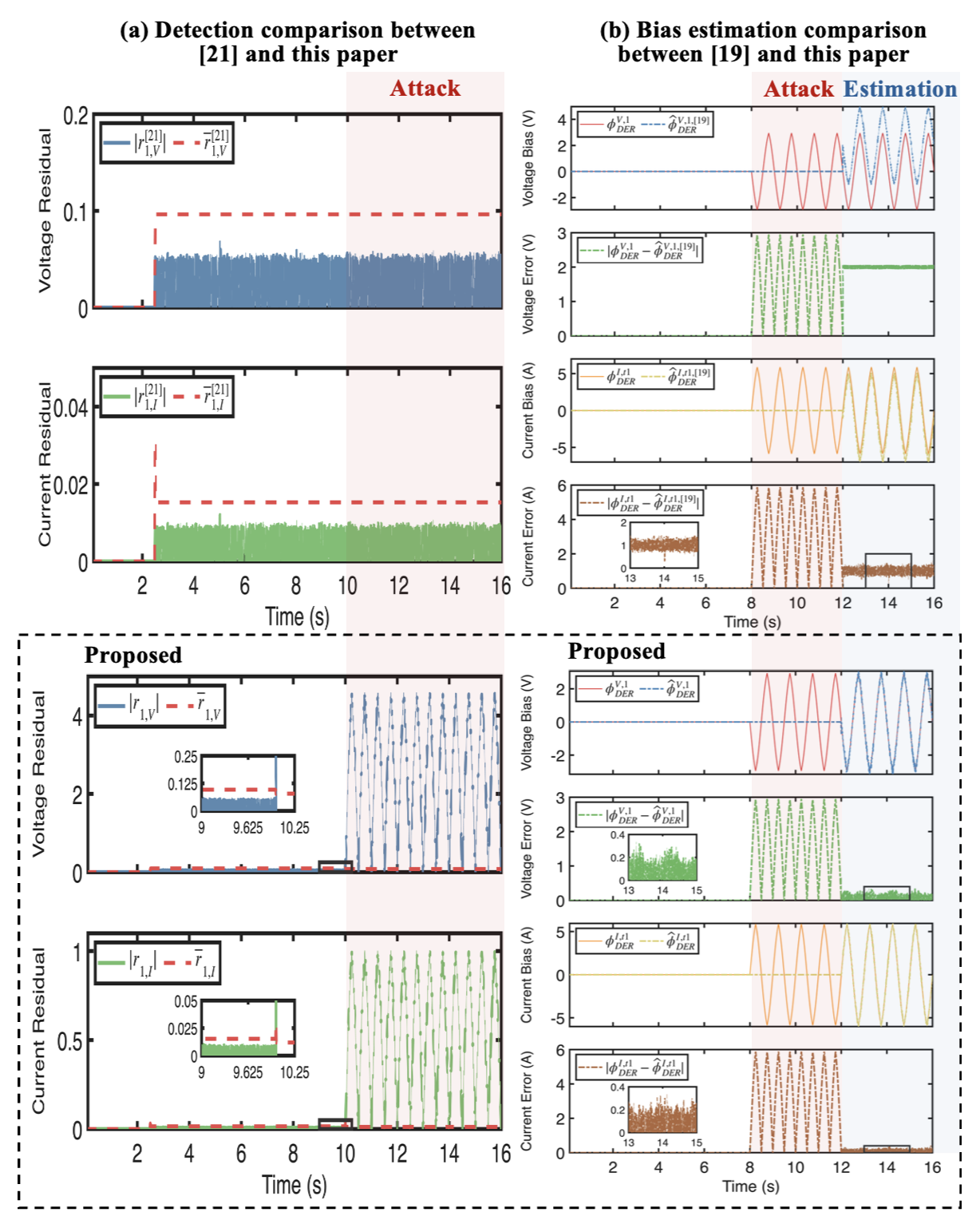}
    \caption{{\color{black}Performance comparison with representative existing detection and mitigation methods. The first column compares the proposed proactive detection with the watermarking-based detection method in \cite{10643207}; the second column compares the bias-estimation performance with the recursive mitigation method in \cite{10746504}. 
    % For the bias-estimation panel, the first and second subplots show the actual/reconstructed voltage bias and the absolute voltage-bias estimation error, while the third and fourth subplots show the corresponding current-bias results.
    }}
    \label{fig:performancecomparison}
\end{figure}

\begin{figure}[!h]
    \centering
    \includegraphics[width=1\linewidth]{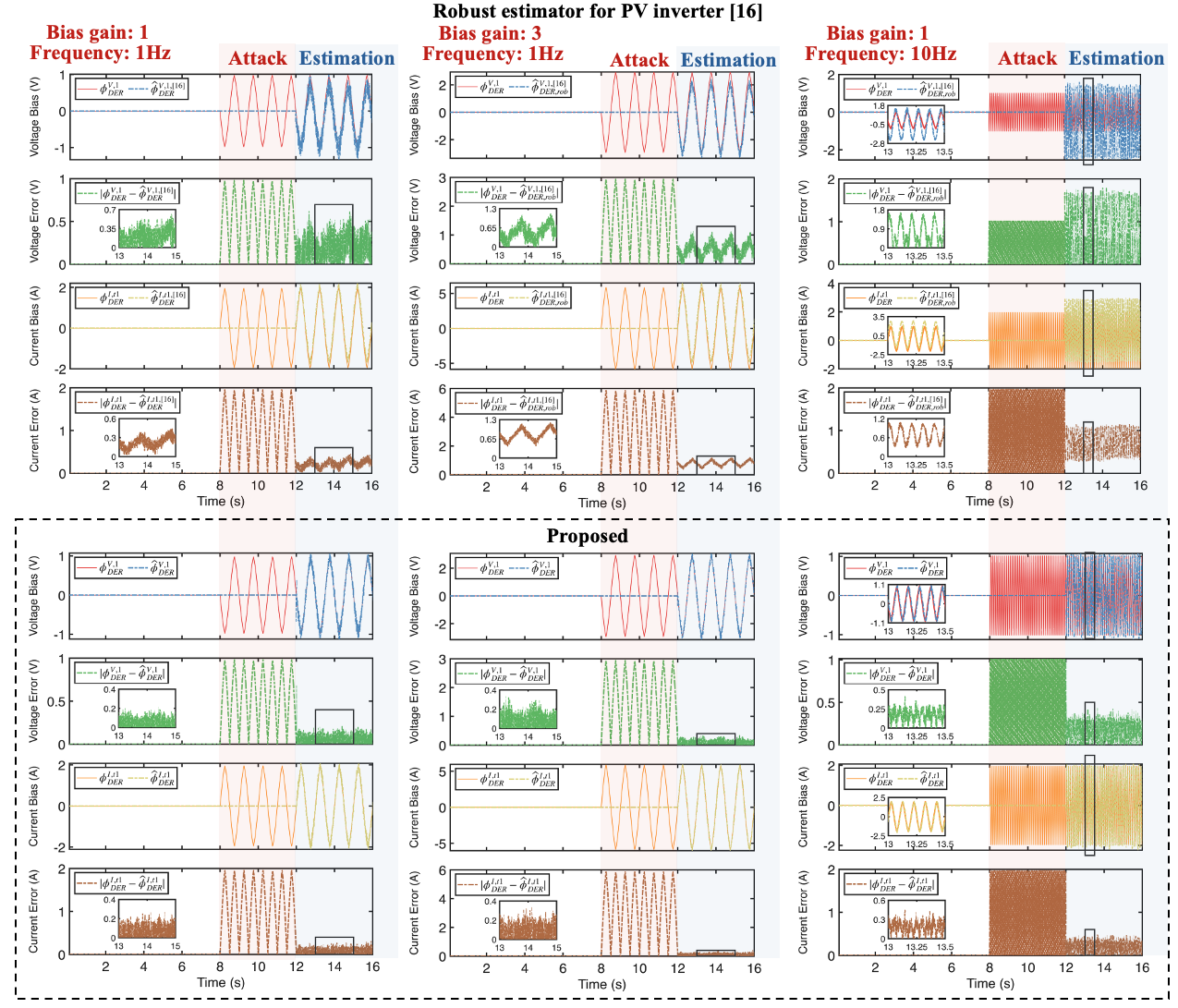}
    \caption{{\color{black}Performance comparison with the robust PV-inverter-style estimator in \cite{10643338} under time-varying spoofing biases with different gains and frequencies. For each bias-estimation panel, the first and second subplots show the actual/reconstructed voltage bias and the absolute voltage-bias estimation error, while the third and fourth subplots show the corresponding current-bias results. The results show that the proposed framework better supports sensor-spoofing mitigation in networked DC microgrids by combining circuit-side perturbation, UIO-based residual generation, and recursive bias reconstruction.}}
    \label{fig:performancecomparison2}
\end{figure}

{\color{black}The robust UIO design is therefore formulated as the following constrained min-max optimisation problem:
\begin{align}
    \min_{\theta} \quad & \max_{\vartheta\in\mathcal{V}} \gamma(\theta,\vartheta) \nonumber\\
    {\rm s.t.}\quad
    & \rho(F_{di}) < \bar{\rho}, \quad
    \|T_{di}(\vartheta)E_{di}(\vartheta)\|_2 \le \epsilon_d,
\end{align}
where $\rho(F_{di})$ enforces the observer stability and $\epsilon_d$ enforces the UIO decoupling condition. The optimisation objective $\gamma(\theta,\vartheta)$ is the finite-horizon robustness metric
\begin{align}
    \gamma(\theta,\vartheta) &= \sqrt{\lambda_{\max}(\Gamma_N(\theta,\vartheta))}, \nonumber\\
    \Gamma_N(\theta,\vartheta) &= \sum_{k=0}^{N}F_{di}^kT_{di}T_{di}^{\rm T}(F_{di}^k)^{\rm T}.
\end{align}
\noindent This metric quantifies the amplification from parameter-mismatch-induced residual disturbances to the estimation error over the uncertainty set $\mathcal{V}$. In the implementation, the above problem is solved by a multi-start derivative-free search. This solver is adopted because the free-structure parameterisation introduces nonlinear spectral-radius and eigenvalue-dependent terms, while the design dimension is low for the considered two-state DER model. A dedicated convex or LMI-based solver that can provide stronger global optimality guarantees will be investigated in future work. With the original UIO parameters, the worst-case metric is $\gamma=14.2134$. After applying the robust UIO design, this value is reduced to $\gamma=2.8995$, corresponding to a $79.6\%$ reduction. This confirms that the robust design significantly attenuates the influence of electrical-parameter variations while preserving the exact decoupling of the unknown input. In Fig. \ref{fig:robustnesstoparametersensitivity}, after the initial detection and reconstruction transient, the robust design confines the voltage-bias estimation error within approximately $0.590$V and the current-bias estimation error within approximately $0.340$A after $t=12$s. Without the robust redesign, the corresponding post-transient error ranges increase to approximately $1.112$V and $0.679$A, respectively. Therefore, the robust UIO design reduces the nontrivial estimation-error envelope induced by electrical-parameter mismatch while keeping the closed-loop power outputs bounded and stable.}

{\color{black}\subsection{Performance Comparison with Existing Methods}
To further clarify the distinction between the proposed framework and representative existing methods, Figs. \ref{fig:performancecomparison} and \ref{fig:performancecomparison2} compare both the detection performance and the bias-estimation performance under sensor spoofing attacks. In Fig. \ref{fig:performancecomparison}, the first column shows that the physics-aware watermarking detector in \cite{10643207} does not trigger under the considered sensor spoofing attack, while the proposed detector identifies the attack immediately after the defence is activated at $t=10$s. The reason is that the spoofing biases are injected before data transmission and are therefore carried together with the watermarked measurements, so the added and removed watermarks can remain mutually consistent. By contrast, the proposed method perturbs the circuit-side electrical relation, proactively forcing the injected biases to deviate from the perturbed physical dynamics, thereby enlarging UIO residuals.}

{\color{black}The second column of Fig. \ref{fig:performancecomparison} shows that the recursive impact mitigation method in \cite{10746504} cannot accurately reconstruct simultaneous voltage and current spoofing biases in the present DC microgrid case. In the bias-estimation setting, the sensor spoofing attack starts at $t=8$s, the DER-layer resistance is perturbed at $t=10$s, and the recursive bias reconstruction starts at $t=12$s. In the upper panels, the voltage and current error subplots remain non-negligible after the attack starts, showing that simultaneous bias reconstruction is inaccurate. In contrast, the lower panels show that the proposed method reduces both voltage- and current-bias estimation errors after the reconstruction is activated; in the zoomed interval $t\in[13,15]$s, both the voltage- and current-error axes are within $[0,0.4]$V and $[0,0.4]$A, respectively. This improvement is obtained because the DER-layer resistance perturbation changes the residual-bias relation and restores sufficient information for recursive bias reconstruction. }

{\color{black}Fig. \ref{fig:performancecomparison2} compares the proposed bias estimation with the robust estimation and defence strategy in \cite{10643338} under time-varying spoofing biases with different gains and frequencies, corresponding to the last three columns of the original comparison figure. The PV-inverter-style robust estimator is adapted to the DC microgrid replay data by constructing a one-step converter residual from the measured voltage/current, load current, neighbouring-DER current, duty cycle, and nominal electrical parameters. Then, a residual-to-bias calibration is introduced to approximately convert the voltage/current residual information into the corresponding spoofing-bias estimates. The error subplots show that the robust estimator can follow part of the bias trend, especially in the lower-frequency cases, but the voltage- and current-bias estimation errors increase when the gain or frequency changes. In the zoomed windows, the robust-estimator error axes expand from approximately $[0,0.7]$V and $[0,0.6]$A in the gain-$1$, $1$Hz case to about $[0,1.3]$V/A in the gain-$3$, $1$Hz case, and further to about $[0,1.8]$V and $[0,1.2]$A in the gain-$1$, $10$Hz case. Quantitatively, over the zoomed interval $t\in[13,15]$s, the robust estimator gives mean/maximum voltage errors of $0.194/0.611$V, $0.491/1.113$V, and $0.747/1.780$V in the three cases, respectively; the corresponding current errors are $0.222/0.438$A, $0.795/1.184$A, and $0.720/1.135$A. By comparison, the proposed method keeps the zoomed post-transient error windows much smaller, mostly within $[0,0.4]$V/A and only enlarged to approximately $[0,0.6]$V/A in the high-frequency case. Its mean/maximum voltage errors in the same three cases are $0.040/0.205$V, $0.068/0.329$V, and $0.173/0.426$V, while the current errors are $0.056/0.264$A, $0.072/0.328$A, and $0.180/0.469$A. This smaller error envelope is achieved because the proposed method integrates the circuit perturbation, UIO residual dynamics, and recursive bias dynamics associated with the hierarchical MG-DER control architecture. It is noted that this comparison does not diminish the value of \cite{10643338}, which has been validated for PV inverter applications, but highlights the additional modelling requirements for mitigation in networked DC microgrids.}

\section{Conclusion}
This paper proposes a hierarchical attack detection and impact mitigation framework to assure the safe operation of NMGs under sensor spoofing attacks, where only the sensors at local PCCs require physical preventive protection. In particular, the data verification at MG layer, centring on the Kirchhoff current law at local PCC, acts as a trigger for the strategic parameter perturbation at DER layer, enhancing detection capability while ensuring circuit stability. By appropriately integrating the MG-DER layer information, the sensor biases could be recursively and accurately estimated without requiring extra hardware installation.

{\color{black}Experimental results validate the effectiveness of the proposed framework across multiple operating and attack scenarios. The resistance perturbation enhances the detectability of stealthy sensor spoofing attacks and enables accurate recursive bias reconstruction, while {its power-loss cost indicates the need for future perturbation-ratio optimisation rather than unbounded resistance increase}. The studies with delayed MG-DER coordination, boost-converter dynamics, load variations, and electrical-parameter uncertainties further show that the framework remains applicable beyond the nominal buck-converter setting, although coordination latency and parameter mismatch can degrade reconstruction accuracy. Comparative results against representative existing methods also demonstrate the benefit of combining circuit-side perturbation, UIO-based residual generation, and hierarchical recursive mitigation for networked DC microgrids. When recursive bias estimation becomes infeasible due to insufficient observability, the open-loop central regulation strategy can still constrain attack propagation at the expense of control performance. Future work will investigate systematic resistance-perturbation optimisation, delay-tolerant mitigation schemes, dedicated robust-UIO solvers with stronger optimality guarantees, and defence strategies for more general sensor spoofing attacks under limited observability.}

% \cite{9200511}.

% Experimental results verify that the triggered resistance perturbation could effectively enhance the detection capability against intelligent sensor spoofing attacks, with side effects of moderately increased power line loss. Qualitatively, a larger resistance perturbation can bring in better data reconstruction accuracy but with growing power line loss. Moreover, despite the asynchronous execution rate in MG and DER layers, the overall defence performance is still promising, supporting the realistic applicability of proposed framework. Future works will focus on applying physics-informed neural networks to improve the proposed framework's adaptability to varying cyber threats and environmental factors \cite{9200511}.

\begin{spacing}{1}
{\footnotesize
\bibliographystyle{IEEEtran} % Choose a style (e.g., plain, alpha, IEEEtran)
\bibliography{root}}

@article{chen1996design,
  title={Design of unknown input observers and robust fault detection filters},
  author={Chen, Jie and Patton, Ron J and Zhang, Hong-Yue},
  journal={International Journal of control},
  volume={63},
  number={1},
  pages={85--105},
  year={1996},
  publisher={Taylor \& Francis}
}

@ARTICLE{gao2016robustuio,
  author={Gao, Zhiwei and Liu, Xiaoxu and Chen, Michael Z. Q.},
  journal={IEEE Transactions on Industrial Electronics},
  title={\textcolor{black}{Unknown Input Observer-Based Robust Fault Estimation for Systems Corrupted by Partially-Decoupled Disturbances}},
  year={2016},
  volume={63},
  number={4},
  pages={2537--2547},
  doi={10.1109/TIE.2015.2497201}
}

@ARTICLE{9793599,
  author={Liu, Mengxiang and Zhao, Chengcheng and Deng, Ruilong and Cheng, Peng and Chen, Jiming},
  journal={IEEE Transactions on Control of Network Systems}, 
  title={False Data Injection Attacks and the Distributed Countermeasure in DC Microgrids}, 
  year={2022},
  volume={9},
  number={4},
  pages={1962-1974},
  doi={10.1109/TCNS.2022.3181483}}

@ARTICLE{10746504,
  author={Liu, Mengxiang and Zhang, Xin and Zhang, Rui and Zhou, Zhuoran and Zhang, Zhenyong and Deng, Ruilong},
  journal={IEEE Transactions on Smart Grid}, 
  title={Detection-Triggered Recursive Impact Mitigation against Secondary False Data Injection Attacks in Cyber-Physical Microgrids}, 
  year={2024},
  volume={},
  number={},
  pages={1-1},
  doi={10.1109/TSG.2024.3493754}}

@inproceedings{tan2015dc,
  title={DC-DC converter modeling and simulation using state space approach},
  author={Tan, Rodney HG and Hoo, Landon YH},
  booktitle={2015 IEEE Conference on Energy Conversion (CENCON)},
  pages={42--47},
  year={2015},
  organization={IEEE}
}

@ARTICLE{10643207,
  author={Liu, Mengxiang and Zhang, Xin and Zhu, Hengye and Zhang, Zhenyong and Deng, Ruilong},
  journal={IEEE Transactions on Information Forensics and Security}, 
  title={Physics-Aware Watermarking Embedded in Unknown Input Observers for False Data Injection Attack Detection in Cyber-Physical Microgrids}, 
  year={2024},
  volume={19},
  number={},
  pages={7824-7840},
  doi={10.1109/TIFS.2024.3447235}}

@INPROCEEDINGS{5941844,
  author={Kaur, Mandeep and Kakar, Shikha and Mandal, Danvir},
  booktitle={2011 3rd International Conference on Electronics Computer Technology}, 
  title={Electromagnetic interference}, 
  year={2011},
  volume={4},
  number={},
  pages={1-5},
  doi={10.1109/ICECTECH.2011.5941844}}

@inproceedings{barua2022halc,
  title={Halc: A real-time in-sensor defense against the magnetic spoofing attack on hall sensors},
  author={Barua, Anomadarshi and Faruque, Mohammad Abdullah Al},
  booktitle={Proceedings of the 25th International Symposium on Research in Attacks, Intrusions and Defenses},
  pages={185--199},
  year={2022}
}

@ARTICLE{10707330,
  author={Wang, Zhiyun and Pan, Kaikai and Xu, Wenyuan},
  journal={IEEE Transactions on Industrial Informatics}, 
  title={Sensor Attacks on Grid-Tie Photovoltaic Inverters: Synthetic Analysis and Real-Time Robust Detection}, 
  year={2025},
  volume={21},
  number={1},
  pages={820-829},
  doi={10.1109/TII.2024.3463704}}

@ARTICLE{10643338,
  author={Pan, Kaikai and Wang, Zhiyun and Dong, Jingwei and Palensky, Peter and Xu, Wenyuan},
  journal={IEEE Transactions on Industrial Electronics}, 
  title={Real-Time Estimation and Defense of PV Inverter Sensor Attacks With Hardware Implementation}, 
  year={2025},
  volume={72},
  number={3},
  pages={3228-3232},
  doi={10.1109/TIE.2024.3436516}}

@ARTICLE{10638139,
  author={Peng, Sha and Liu, Mengxiang and Chai, Li and Deng, Ruilong},
  journal={IEEE Transactions on Smart Grid}, 
  title={DST-GNN: A Dynamic Spatiotemporal Graph Neural Network for Cyberattack Detection in Grid-Tied Photovoltaic Systems}, 
  year={2025},
  volume={16},
  number={1},
  pages={330-343},
  doi={10.1109/TSG.2024.3445113}}

@ARTICLE{9580468,
  author={Zhang, Jinan and Guo, Lulu and Ye, Jin},
  journal={IEEE Transactions on Smart Grid}, 
  title={Cyber-Attack Detection for Photovoltaic Farms Based on Power-Electronics-Enabled Harmonic State Space Modeling}, 
  year={2022},
  volume={13},
  number={5},
  pages={3929-3942},
  doi={10.1109/TSG.2021.3121009}}

@ARTICLE{9621221,
  author={Liu, Mengxiang and Zhao, Chengcheng and Zhang, Zhenyong and Deng, Ruilong and Cheng, Peng and Chen, Jiming},
  journal={IEEE Transactions on Smart Grid}, 
  title={Converter-Based Moving Target Defense Against Deception Attacks in DC Microgrids}, 
  year={2022},
  volume={13},
  number={5},
  pages={3984-3996},
  doi={10.1109/TSG.2021.3129195}}

@inproceedings{BH2025solar,
  title={A Closer Look at the Gaps in the Grid: New Vulnerabilities and Exploits Affecting Solar Power Systems},
  author={Daniel, Santos and Francesco, La Spina and Stanislav, Dashevskyi},
  booktitle={Blackhat Aisa},
  pages={},
  year={2025}
}

@ARTICLE{11029084,
  author={Zhang, Suhan and Zhang, Xin and Zhang, Rui and Gu, Wei and Cao, Ge},
  journal={IEEE Transactions on Smart Grid}, 
  title={N-1 Evaluation of Integrated Electricity and Gas System Considering Cyber-Physical Interdependence}, 
  year={2025},
  volume={16},
  number={5},
  pages={3728-3742},
  doi={10.1109/TSG.2025.3578271}}

@misc{Forescout2025,
  author       = "{Forescout Research – Vedere Labs}",
  title        = "{SUN:DOWN – Destabilizing the Grid via Orchestrated Exploitation of Solar Power Systems}",
  year         = {2025},
  institution  = "{Forescout Technologies, Inc.}",
  url          = {https://www.forescout.com/resources/sun-down-research-report/},
  note         = "Accessed: 2025-03-28"
}

@inproceedings{barua2020hall,
  title={Hall Spoofing: A Non-Invasive DoS Attack on Grid-Tied Solar Inverter},
  author={Barua, Anomadarshi and Al Faruque, Mohammad Abdullah},
  booktitle={29th USENIX Security Symposium (USENIX Security 20)},
  pages={1273--1290},
  year={2020}
}

@inproceedings{yang2024rethink,
  title={\textcolor{black}{ReThink: Reveal the Threat of Electromagnetic Interference on Power Inverters}},
  author={Yang, Fengchen and Dan, Zihao and Pan, Kaikai and Yan, Chen and Ji, Xiaoyu and Xu, Wenyuan},
  booktitle={\textcolor{black}{Proceedings of the Network and Distributed System Security (NDSS) Symposium}},
  year={\textcolor{black}{2025}}
}

@inproceedings{tu2021transduction,
  title={Transduction shield: A low-complexity method to detect and correct the effects of EMI injection attacks on sensors},
  author={Tu, Yazhou and Tida, Vijay Srinivas and Pan, Zhongqi and Hei, Xiali},
  booktitle={Proceedings of the 2021 ACM Asia Conference on Computer and Communications Security},
  pages={901--915},
  year={2021}
}

@ARTICLE{9796617,
  author={Gu, Yunjie and Green, Timothy C.},
  journal={Proceedings of the IEEE}, 
  title={Power System Stability With a High Penetration of Inverter-Based Resources}, 
  year={2023},
  volume={111},
  number={7},
  pages={832-853},
  doi={10.1109/JPROC.2022.3179826}}

@ARTICLE{10459229,
  author={Liu, Mengxiang and Teng, Fei and Zhang, Zhenyong and Ge, Pudong and Sun, Mingyang and Deng, Ruilong and Cheng, Peng and Chen, Jiming},
  journal={IEEE Transactions on Smart Grid}, 
  title={Enhancing Cyber-Resiliency of DER-Based Smart Grid: A Survey}, 
  year={2024},
  volume={15},
  number={5},
  pages={4998-5030},
  doi={10.1109/TSG.2024.3373008}}

@ARTICLE{9737024,
  author={Liu, Zengji and Wang, Qi and Ye, Yujian and Tang, Yi},
  journal={IEEE Transactions on Smart Grid}, 
  title={A GAN-Based Data Injection Attack Method on Data-Driven Strategies in Power Systems}, 
  year={2022},
  volume={13},
  number={4},
  pages={3203-3213},
  doi={10.1109/TSG.2022.3159842}}

@ARTICLE{9583906,
  author={Deng, Chao and Guo, Fanghong and Wen, Changyun and Yue, Dong and Wang, Yu},
  journal={IEEE Transactions on Industrial Electronics}, 
  title={Distributed Resilient Secondary Control for DC Microgrids Against Heterogeneous Communication Delays and DoS Attacks}, 
  year={2022},
  volume={69},
  number={11},
  pages={11560-11568},
  doi={10.1109/TIE.2021.3120492}}

@ARTICLE{9832494,
  author={Yao, Weitao and Wang, Yu and Xu, Yan and Deng, Chao},
  journal={IEEE Transactions on Industrial Informatics}, 
  title={Cyber-Resilient Control of an Islanded Microgrid Under Latency Attacks and Random DoS Attacks}, 
  year={2023},
  volume={19},
  number={4},
  pages={5858-5869},
  doi={10.1109/TII.2022.3191315}}

@ARTICLE{10836761,
  author={Liu, Mengxiang and Zhang, Xin and Zhao, Chengcheng and Deng, Ruilong},
  journal={IEEE Transactions on Power Systems}, 
  title={Matrix Coding Enabled Impact Mitigation against Primary False Data Injection Attacks in Cyber-Physical Microgrids}, 
  year={2025},
  volume={},
  number={},
  pages={1-16},
  doi={10.1109/TPWRS.2025.3528322}}

@ARTICLE{9069415,
  author={Zhang, Zhuhaobo and Zhu, Fan and Xu, Dehong and Krein, Philip T. and Ma, Hao},
  journal={IEEE Transactions on Power Electronics}, 
  title={An Integrated Inductive Power Transfer System Design With a Variable Inductor for Misalignment Tolerance and Battery Charging Applications}, 
  year={2020},
  volume={35},
  number={11},
  pages={11544-11556},
  doi={10.1109/TPEL.2020.2987906}}

@ARTICLE{liu2011hybrid,
  author={Liu, Xiong and Wang, Peng and Loh, Poh Chiang},
  journal={IEEE Transactions on Smart Grid},
  title={\textcolor{black}{A Hybrid AC/DC Microgrid and Its Coordination Control}},
  year={2011},
  volume={2},
  number={2},
  pages={278-286},
  doi={10.1109/TSG.2011.2116162}}

@ARTICLE{yan2019smallsignal,
  author={Yan, Yimajian and Shi, Di and Bian, Desong and Huang, Bibin and Yi, Zhehan and Wang, Zhiwei},
  journal={IEEE Transactions on Smart Grid},
  title={\textcolor{black}{Small-Signal Stability Analysis and Performance Evaluation of Microgrids Under Distributed Control}},
  year={2019},
  volume={10},
  number={5},
  pages={4848-4858},
  doi={10.1109/TSG.2018.2869566}}

@ARTICLE{sahoo2021ac,
  author={Sahoo, Subham and Yang, Yongheng and Blaabjerg, Frede},
  journal={IEEE Transactions on Power Electronics},
  title={\textcolor{black}{Resilient Synchronization Strategy for AC Microgrids Under Cyber Attacks}},
  year={2021},
  volume={36},
  number={1},
  pages={73-77},
  doi={10.1109/TPEL.2020.3005208}}

@ARTICLE{he2022lowinertia,
  author={He, Changjun and He, Xiuqiang and Geng, Hua and Sun, Huadong and Xu, Shiyun},
  journal={IEEE Transactions on Energy Conversion},
  title={\textcolor{black}{Transient Stability of Low-Inertia Power Systems With Inverter-Based Generation}},
  year={2022},
  volume={37},
  number={4},
  pages={2903-2912},
  doi={10.1109/TEC.2022.3185623}}

@ARTICLE{li2022duality,
  author={Li, Yitong and Gu, Yunjie and Green, Timothy C.},
  journal={IEEE Transactions on Power Systems},
  title={\textcolor{black}{Revisiting Grid-Forming and Grid-Following Inverters: A Duality Theory}},
  year={2022},
  volume={37},
  number={6},
  pages={4541-4554},
  doi={10.1109/TPWRS.2022.3151851}}

@ARTICLE{yoo2020hybrid,
  author={Yoo, Hyeong-Jun and Nguyen, Thai-Thanh and Kim, Hak-Man},
  journal={IEEE Transactions on Sustainable Energy},
  title={\textcolor{black}{Consensus-Based Distributed Coordination Control of Hybrid AC/DC Microgrids}},
  year={2020},
  volume={11},
  number={2},
  pages={629-639},
  doi={10.1109/TSTE.2019.2899119}}

@ARTICLE{espina2021hybrid,
  author={Espina, Enrique and Cardenas-Dobson, Roberto and Simpson-Porco, John W. and Saez, Doris and Kazerani, Mehrdad},
  journal={IEEE Transactions on Power Electronics},
  title={\textcolor{black}{A Consensus-Based Secondary Control Strategy for Hybrid AC/DC Microgrids With Experimental Validation}},
  year={2021},
  volume={36},
  number={5},
  pages={5971-5984},
  doi={10.1109/TPEL.2020.3031539}}

@ARTICLE{chang2021interlinking,
  author={Chang, Jae Won and Lee, Gyu Sub and Moon, Seung Il and Hwang, Pyeong Ik},
  journal={IEEE Transactions on Smart Grid},
  title={\textcolor{black}{A Novel Distributed Control Method for Interlinking Converters in an Islanded Hybrid AC/DC Microgrid}},
  year={2021},
  volume={12},
  number={5},
  pages={3765-3779},
  doi={10.1109/TSG.2021.3074706}}

@ARTICLE{abramson2018reconfigurable,
  author={Abramson, Rose A. and Gunter, Samantha J. and Otten, David M. and Afridi, Khurram K. and Perreault, David J.},
  journal={IEEE Transactions on Power Electronics},
  title={\textcolor{black}{Design and Evaluation of a Reconfigurable Stacked Active Bridge DC-DC Converter for Efficient Wide Load Range Operation}},
  year={2018},
  volume={33},
  number={12},
  pages={10428-10448},
  doi={10.1109/TPEL.2018.2801306}}

@ARTICLE{mokhtar2019adaptiveDroop,
  author={Mokhtar, Mohamed and Ibrahim, Ahmed K. and El-Sattar, Ahmed A.},
  journal={IEEE Transactions on Smart Grid},
  title={\textcolor{black}{An Adaptive Droop Control Scheme for DC Microgrids Integrating Sliding Mode Voltage and Current Controlled Boost Converters}},
  year={2019},
  volume={10},
  number={2},
  pages={1685-1693},
  doi={10.1109/TSG.2017.2776281}}

@ARTICLE{wu2017virtualImpedance,
  author={Wu, Xiangyu and Shen, Chen and Iravani, Reza},
  journal={IEEE Transactions on Smart Grid},
  title={\textcolor{black}{Feasible Range and Optimal Value of the Virtual Impedance for Droop-Based Control of Microgrids}},
  year={2017},
  volume={8},
  number={3},
  pages={1242-1251},
  doi={10.1109/TSG.2016.2519454}}

@MISC{digikey_sd500l48,
  author={{DigiKey}},
  title={\textcolor{black}{{SD-500L-48 MEAN WELL USA Inc.}}},
  howpublished={\url{https://www.digikey.com/en/products/detail/mean-well-usa-inc/SD-500L-48/7706541}},
  note={Accessed: 2026-04-24}}

@MISC{digikey_aev20eb,
  author={{DigiKey}},
  title={\textcolor{black}{{AEV20E-B Altran Magnetics, LLC}}},
  howpublished={\url{https://www.digikey.com/en/products/detail/altran-magnetics-llc/AEV20E-B/18739452}},
  note={Accessed: 2026-04-24}}

@MISC{digikey_hs50r1j,
  author={{DigiKey}},
  title={\textcolor{black}{{HS50 R1 J Ohmite}}},
  howpublished={\url{https://www.digikey.com/en/products/detail/ohmite/HS50-R1-J/5307917}},
  note={Accessed: 2026-04-24}}

@ARTICLE{qin2025adaptive,
  author={Qin, Xiangyu and Lin, Zhengyu and Jiang, Wei and Lee, Hazel},
  journal={Energies},
  title={\textcolor{black}{Adaptive Line Resistance Estimation and Compensation for Accurate Power Sharing of Droop-Controlled DC Microgrids}},
  year={2025},
  volume={18},
  number={9},
  pages={2183},
  doi={10.3390/en18092183}}

@ARTICLE{liu2019droop,
  author={Liu, Yuchao and Green, Tim C. and Wu, Jian and Rouzbehi, Kumars and Raza, Ali and Xu, Dianguo},
  journal={IEEE Access},
  title={\textcolor{black}{A New Droop Coefficient Design Method for Accurate Power-Sharing in VSC-MTDC Systems}},
  year={2019},
  volume={7},
  pages={47605-47614},
  doi={10.1109/ACCESS.2019.2909044}}
\end{spacing}

{\appendix
\subsection{\texorpdfstring{Proof of Proposition \ref{propos:residual sensitivity}}{Proof of Proposition}}\label{appendix: proof of propos}
\begin{proof}
Since the sampling period $T_{samp}$ is close to the microsecond level in the converter primary control, the properties of discrete-time system \eqref{eq: Discrete DER SS Model} would largely inherit from the continuous-time system \eqref{eq: DER SS Model}. To simplify the analysis, the subsequent analysis will be based on the continuous-time dynamical forms of DER and UIO, with their relations to the discrete-time forms given by $F_i \sim F_{di}$ and $T_i \sim T_{di}$.
% , and $\tilde{\bm{r}}_i^{\phi} \sim \bm{r}_i^{\phi}$. 
Suppose that there are alternation steps on these electrical parameters, characterised by prefix $\Delta$, then the resulting residual variation can be calculated as
{\small\begin{align}\label{eq: sensitivity midterm}
    {\bm{r}}_i^{\phi}(k+1) = F_i{\bm{r}}_i^{\phi}(k) + T_i(\widetilde{A}_{i} - A_{i})\bm{\phi}_i(k) + T_i\widetilde{E}_{i}g_i^a(k),
\end{align}}where $\widetilde{A}_{i}, \widetilde{E}_{i}$ denote the system parameters after applying the alternations. According to \eqref{eq: system parameters} and \eqref{eq: UIO parameter 1}, $T_i$ should has the form of $[0,\bm{t}_{i2}]$, where $\bm{t}_{i2}$ is a non-zero vector. Then, it is deduced that the electrical parameter alterations that only cause variations on the first row of $A_i$ will not contribute to increasing the UIO residual, i.e., the establishment of \eqref{eq: zero sensitivity}. The last term in \eqref{eq: sensitivity midterm} is always zero due to $T_i\widetilde{E}_i = 0$, which will not be affected by the alteration of $C_{ti}$. Therefore, the residual variation $\bm{r}_i^{\phi}$ is only related to the alterations on electrical parameters $R_{ti}, L_{ti}$. 

In particular, the residual variation would be proportional to the middle term in \eqref{eq: sensitivity midterm}, which, assuming $\bm{\phi}_i = [1, 1]^{\rm T}$, can be approximated as
\begin{align}\label{eq: residual approximation}
    \bm{r}_i^{\phi} \varpropto
    \frac{(R_{ti} + 1)\Delta L_{ti}}{L_{ti}^2} \bm{t}_{i2} - \frac{\Delta R_{ti}}{L_{ti}} \bm{t}_{i2}.
\end{align}Then, \eqref{eq: inductance sensitivity} and \eqref{eq: resistance sensitivity} can be directly derived.
\end{proof}

\subsection{\texorpdfstring{Proof for the Observability of Matrix Pair $(F_i, H_i)$}{Proof for the Observability of Matrix Pair (Fi, Hi)}}\label{appendix: observability proof}
\begin{proof}
The matrix pair $(F_i, H_i)$ is observable if and only if for each $\lambda_i \in \sigma(F_i)$, there is 
\begin{align}\label{eq: observability conditions}
\rm{rank}\begin{bmatrix} 
\lambda_i I - F_i \\
H_i
\end{bmatrix}=3.
\end{align}According to \eqref{eq: EstimationParameter}, the upper triangle matrix $F_i$ has eigenvalues $\sigma(F_i) = \sigma(A_i) \cup \{1\}$, thus informing the division of subsequent proof into two parts: 1) For the eigenvalue $\lambda_i \in \sigma(A_i)$, the establishment of \eqref{eq: observability conditions} is equivalent to the observability of matrix pair $(A_i, \widetilde{T}_{i}(A_{i} - \widetilde{A}_{i}))$, i.e., 
\begin{align}\label{eq: reduced observability conditions}
\rm{rank}\begin{bmatrix} 
\lambda_i I - A_i \\
\widetilde{T}_{i}(A_{i} - \widetilde{A}_{i})
\end{bmatrix}=2,
\end{align}which can be assured if the corresponding eigenvector $\bm{\xi}_i$ satisfies
\begin{align}\label{eq: middle 1}
\widetilde{T}_{i}(A_{i} - \widetilde{A}_{i})\bm{\xi}_i \ne \bm{0}
\end{align}Based on the UIO condition $\widetilde{T}_i\widetilde{E}_i = \bm{0}$, \eqref{eq: middle 1} can be further transformed into 
\begin{align}\label{eq: middle 2}
(A_{i} - \widetilde{A}_{i})\bm{\xi}_i \ne \alpha \widetilde{E}_i, \alpha \ne 0.
\end{align}Since the resistance perturbation only alters entry $a_{22}$ of $A_i$, and vector $\widetilde{E}_i$ has only one non-zero element at the first entry, we have $(A_i - \widetilde{A}_i)\widetilde{E}_i = \bm{0}$, indicating that \eqref{eq: middle 2} is assured.

2) For eigenvalue $\lambda_i = 1$, \eqref{eq: observability conditions} means that the corresponding eigenvector $[\bm{\xi}_i^{\rm T}, \gamma_i]$ needs to ensure that when $(I-A_i)\bm{\xi}_i + E_i\gamma_i = \bm{0}$, $\widetilde{T}_i(A_i - \widetilde{A}_i)\bm{\xi}_i + \widetilde{T}_iE_i\gamma_i \ne \bm{0}$ is always satisfied. Substituting $\bm{\xi_i} = -(I-A_i)^{-1}E_i\gamma_i$ into the latter inequality, it is transformed into
\begin{align}\label{eq: middle 3}
\widetilde{T}_i((A_i - \widetilde{A}_i)(I-A_i)^{-1} + I)E_i \ne \bm{0}
\end{align}Since $\widetilde{T}_iE_i = \widetilde{T}_i\widetilde{E}_i = \bm{0}$, while the middle matrix $(A_i - \widetilde{A}_i)(I-A_i)^{-1} + I$ is going to change the direction of $E_i$, \eqref{eq: middle 3} always stands.
\end{proof}}

% \begin{figure*}[!h]
%     \centering
%     \includegraphics[width=1\linewidth]{Picture/SensorSpoofingAttackIllustration.jpg}
%     \caption{This figure illustrates the mechanisms of spoofing attacks against voltage and current sensors in DERs \cite{yang2024rethink}. In the left part, the parasitic capacitance carried on the sensor's PCB couples the high-frequency electric fields resulting from EMI injections into the differential operational amplifier (OP-AMP). Then, the coupled interference is rectified and amplified by OP-AMP, eventually outputted as positive or negative bias on the voltage reading. In the right part, the current sensor has not only the OP-AMP circuit but also a Hall element, which is a new entrance for EMI injections. Specifically, the induced magnetic ($B_A$) or electric ($E_A$) field around the Hall chip is converted as voltage deviation $V_A$, which, after the process of OP-AMP circuit, is outputted as positive or negative bias on the current reading.}
%     \label{fig:sensorspoofing}
% \end{figure*}

\end{spacing}

\end{document}